\documentclass[a4paper,12pt]{article}
\usepackage{amssymb}
\usepackage{amsfonts}
\usepackage{amsmath}
\usepackage{geometry}
\usepackage{setspace}
\usepackage[round]{natbib}
\usepackage[pdftex]{graphicx}
\usepackage{epstopdf}
\usepackage{framed}
\usepackage{rotating,booktabs}
\usepackage{color}
\usepackage{algorithm}
\usepackage{multirow}
\usepackage{caption}
\usepackage{subcaption}
\usepackage{algpseudocode}
\newtheorem{theorem}{Theorem}

\newenvironment{proof}[1][Proof]{\textbf{#1.} }{\ \rule{0.5em}{0.5em}}

\begin{document}
\title{\textbf{Quantifying Omitted Variable Bias in Nonlinear Instrumental Variable Estimators}\thanks{We thank Carlos Cinelli and Ting-Yu Kuo for constructive discussions and seminar participants in 2024 Annual Meeting of Taiwan Econometric Society (National Taiwan Normal University), 2024 Macroeconometric Modelling Workshop (Academia Sinica), the 8th International Conference on Econometrics and Statistics (EcoSta 2025, Waseda University) and National Taiwan University for helpful comments.}}
\author{Yu-Min Yen\thanks{Department of International Business, National Chengchi University, 64, Section 2, Zhi-nan Road, Wenshan, Taipei 116, Taiwan. E-mail: \texttt{yyu\_min@nccu.edu.tw.}}}
\date{\today}
\maketitle
\begin{abstract}
We develop a framework for quantifying omitted variable bias (OVB) in nonlinear instrumental variable (IV) estimators, including the local average treatment effect (LATE), the LATE for the treated (LATT), and the partially linear IV model (PLIVM). Extending sensitivity analysis beyond linear settings, we derive bias decompositions, establish partial identification bounds, and construct OVB-adjusted confidence intervals. We estimate OVB bounds and conduct inference using double machine learning (DML), allowing flexible control for high-dimensional covariates. An application to the U.S. Job Training Partnership Act (JTPA) experiment shows that, at conventional significance levels, first-stage compliance estimates are robust to omitted variables, whereas intention-to-treat and treatment effects are more sensitive. Program impacts are robust and significant for females but fragile for males.	
\\\\
\textbf{Keywords: }Causal inference, Machine learning, Microeconometris, Sensitivity analysis, Partial identification\\
\end{abstract}
\clearpage

\onehalfspacing
\section{Introduction}
Instrumental variable (IV) estimators are widely used to address endogeneity, but their validity is compromised when relevant variables are omitted. This paper extends the results of \citet{CCNSS_2024} by quantifying omitted variable bias (OVB) in a broad class of IV estimators, including the partially linear IV model (PLIVM) and nonlinear estimators such as the local average treatment effect (LATE) and the local average treatment effect for the treated (LATT). 
Let $Z$ denote the instrumental variable, $X$ a set of observable covariates and $A$ a set of unobservable (or omitted) covariates. Define $W:=(Z,X,A)$ and $W_{s}:=(Z,X)$. Let $Y$
denote the outcome (dependent variable) and $D$ the treatment, which may be endogenous. Many IV estimators can be written in the following form: 
\begin{equation}
	\theta=\frac{\lambda}{\gamma},\label{main_estimand_long_version}
\end{equation}where
\[\lambda = E[\alpha(W)g_{Y}(W)] \text{ and }\gamma =E[\alpha(W)g_{D}(W)].
\]
Inside the above expectations, $\alpha(W)$ is a weighting function, $g_{Y}(W)=E[Y|W]$ and
$g_{D}(W)=E[D|W]$. When only $W_s$ is available, the short version of (\ref{main_estimand_long_version}) is given by 
\begin{equation}
	\theta_{s}=\frac{\lambda_s}{\gamma_s},\label{main_estimand_short_version}
\end{equation}where
\[
\lambda_s = E[\alpha_{s}(W_{s})g_{Ys}(W_{s})]\text{ and }\gamma_s = E[\alpha_{s}(W_{s})g_{Ds}(W_{s})]
\]
Again, inside the above expectations, $\alpha_{s}(W_{s})$ is a weighting and $g_{Ys}(W_{s})=E[Y|W_{s}]$
and $g_{Ds}(W_{s})=E[D|W_{s}]$. \citet{CCNSS_2024} refer to $\alpha(W)$ and $\alpha_{s}(W_s)$ as the long and short \textit{Riesz representers (RR)}. 

We assume that with $W$, $\theta$ correctly identifies the interested parameter. However, with $W_s$, $\theta_s$ in general does not correctly identify the interested parameter. The central object of this paper are to characterize magnitude of the omitted variable bias (OVB) caused by using $\theta_s$:
\[
|\text{Bias}| := |\theta - \theta_s|,
\]
and construct the OVB bounds for partially identifying $\theta$. We also will develop relevant statistical inference tools for the OVB analysis.

Many IV estimators admit the form
of (\ref{main_estimand_long_version}), and we give several examples
as follows.

\textbf{Example 1 (PLIVM)}: Consider the omitted variable bias (OVB) of a partially linear instrumental variable (IV) regression. We assume that the dependent
variable $Y$ and endogenous variable $D$ (can be continuous) have
a form of partial linearity as follows:
\begin{align}
	Y & =\theta D+f(X,A)+u_{Y},\label{Y_plm}\\
	D & =\gamma Z+h(X,A)+u_{D}.\label{D_plm}
\end{align}
Our goal is to estimate the coefficient $\theta$. Following \citet{CCNSS_2024}, we will refer (\ref{Y_plm}) and (\ref{D_plm}) as long versions of $Y$ and $D$ (since they are constructed with a completed set of variables $W$). Assume $E[u_{Y}|W]=0$ and $E[u_{D}|W]=0$, but $E[u_{Y}|D]\neq0$ and there is endogeneity. When endogeneity is present, $\theta$ can be identified with a two-stage procedure. At first, we rewrite $Y$ as a reduced form
\[
Y=\lambda Z+k(X,A)+\varepsilon_{Y},
\]
where $\lambda=\theta\gamma$, $k(X,A)=\theta h(X,A)+f(X,A)$ and
$\varepsilon_{Y}=\theta u_{D}+u_{Y}.$ Note that $E[\varepsilon_{Y}|W]=0$, and then $\lambda$ and $\gamma$ can be identified as
\[
\lambda =E[\alpha(W)g_{Y}(W)], \gamma =E[\alpha(W)g_{D}(W)],
\]
where
\begin{align*}
	\alpha(W) & =\frac{Z-E[Z|X,A]}{E[(Z-E[Z|X,A])^{2}]},
\end{align*}
$g_{Y}(W)=E[Y|W]$ and $g_{D}(W)=E[D|W]$. Then we can have 
\[
\theta=\frac{\lambda}{\gamma}=\frac{E[\alpha(W) g_{Y}(W)]}{E[\alpha(W) g_{D}(W)]},
\]given that $\gamma\neq 0$.
With $W_{s}:=(Z,X)$, the short versions of $Y$ and $D$ in (\ref{Y_plm}) and (\ref{D_plm}) are given by
\begin{align}
	Y & =\theta_{s}D+f_{s}(X)+u_{Ys},\label{Y_plm_s}\\
	D & =\gamma_{s}Z+h_{s}(X)+u_{D_{s}}.\label{D_plm_s}
\end{align}
Following similar two-stage procedure above, the $\theta_{s}$ in short version (\ref{Y_plm_s}) can be identified as
\[
\theta_{s}=\frac{\lambda_{s}}{\gamma_{s}}=\frac{E[\alpha_{s}(W_{s})g_{Ys}(W_{s})]}{E[\alpha_{s}(W_{s})g_{Ds}(W_{s})]},
\]
given that $\gamma_s\neq 0$, where 
\[
\alpha_{s}(W_{s})=\frac{Z-E[Z|X]}{E[(Z-E[Z|X])^{2}]},
\]
and $g_{Ys}(W_{s})=E[Y|W_{s}]$ and $g_{Ds}=E[D|W_{s}]$. 

\textbf{Example 2 (LATE)}: Consider the two-sided non-compliance framework. Let $Y_d$ denote the potential outcomes when the treatment variable $D=d$. Let $T\in \{AT,NT,C\}$ denote type of an individual, where $AT$ refer to the always taker, $NT$ refer to the never taker and $C$ refer to the complier. Under certain assumptions \citep{Frolich_2007}, the local average treatment effect (average treatment effect of the complier group, LATE, \citet{IA_1994}) \[\text{LATE}:= E[Y_{1}-Y_{0}|T=C]\] can be identified as
\begin{equation}
	\theta = \frac{\text{ITT}}{P(T=C)}=\frac{\lambda}{\gamma}=\frac{E[\alpha(W)g_{Y}(W)]}{E[\alpha(W)g_{D}(W)]},
	\label{LATE}
\end{equation}
where ITT is the intention-to-treat effect and $P(T=C)$ is the probability that an individual is a complier. In the expectations, the weight function $\alpha(W)$ is given by:
\[
\alpha(W) =\frac{Z}{\pi(X,A)}-\frac{1-Z}{1-\pi(X,A)},
\]where $\pi(X,A) =P(Z=1|X,A)$ is the propensity score function of the instrumental variable, and $g_{Y}(W)=E[Y|W]$ and $g_{D}(W)=E[D|W]$. Equation (\ref{LATE}) is a ratio of two inverse propensity score weighting estimands. The short version of $\theta$ is given by
\begin{equation}
	\theta_{s} =\frac{\lambda_s}{\gamma_s}=\frac{E\left[\alpha_{s}(W_{s})g_{Ys}(W_{s})\right]}{E[\alpha_{s}(W_{s})g_{Ds}(W_{s})]},\label{LATE_s}
\end{equation}
where
\[\alpha_{s}(W_{s}) =\frac{Z}{\pi_{s}(X)}-\frac{1-Z}{1-\pi_{s}(X)}, \pi_{s}(X) =P(Z=1|X),\]
and $g_{Ys}(W_{s})=E[Y|W_{s}]$ and $g_{Ds}(W)=E[D|W_{s}]$.

\textbf{Example 3 (LATT)}: The local average treatment effect on the treated (LATT) is defined as
\[
\text{LATT} := E[Y_1-Y_0|T=C,D=1].
\]That is, LATE for the treated compliers. Under the same assumptions for identifying LATE, LATT can also be identified as $\theta = \lambda/\gamma$. But the weight functions $\alpha(W)$ and $\alpha(W_s)$ for LATT become
\begin{eqnarray}
	\alpha(W) 
	&=&\frac{1}{P_Z}\left[Z-\frac{\pi(X,A)}{1-\pi(X,A)}(1-Z)\right],\label{LATT_alpha}\\
	\alpha_{s}(W_{s}) &=&\frac{1}{P_Z}\left[Z-\frac{\pi_{s}(X)}{1-\pi_{s}(X)}(1-Z)\right],\label{LATT_alphas}
\end{eqnarray}where $P_Z = P(Z=1)$.
The function $g_Y(W)$, $g_{Y_s}(W_s)$,  $g_D(W)$ and $g_{D_s}(W_s)$ are the same as in LATE. 

The rest of the paper is organized as follows. Section 2 introduces the proposed method for OVB analysis of IV estimators, including the construction of OVB bounds, a set of statistical inference tools and a method for estimating them using double machine learning (DML). Section 3 applies the method to an empirical analysis of LATE and LATT using the classical JTPA data. Section 4 concludes.

\section{Methodology}
\subsection{The OVB Bounds for $\lambda$ and $\gamma$}
To quantify the OVB of $\theta_{s}\in \boldsymbol{\Theta}_s$, instead of directly calculating the OVB through comparing $\theta\in \boldsymbol{\Theta}$ and $\theta_s$, we exploit results from inference with weak instrument variables (section 13.3 in \citet{CHKSS_2024}). A similar strategy, based on the Anderson–Rubin regression, was also adopted by \citet{CH_2025} to construct the OVB bound in a linear IV model. Suppose we would like to test $H_{0}:\theta=\theta_0$, where $\theta_0\in \boldsymbol{\Theta}_0\subseteq \boldsymbol{\Theta}$. Let $\phi_{t}:=\lambda-\gamma t$ and
testing $H_{0}$ is equivalent to testing $H_{0}^{\prime}:\phi_{\theta_0}=0$. We next show that $\phi_{t}$ can be partially identified when the short version estimands $\lambda_s$ and $\gamma_s$ are used. To simplify the notations, we use $(\alpha,\alpha_s, g_Y, g_{Ys}, g_D, g_{Ds})$ to replace $(\alpha(W),\alpha_s(W_s), g_Y(W), g_{Ys}(W_s), g_D(W), g_{Ds}(W_s))$ in the following discussion. 

At first, bias of $\lambda_{s}$ can be expressed as:
\begin{align}
	\lambda-\lambda_{s} & =E[(\alpha-\alpha_{s})(g_{Y}-g_{Ys})],\label{OVB_lambda}
\end{align}
by using the result in \citet{CCNSS_2024}. A key condition to achieving (\ref{OVB_lambda}) is that $E[g_{Y_s}(\alpha-\alpha_s)]=0$, which holds for LATE, LATT and PLIVM. With some calculations, we can have the following result for the squared bias of $\lambda_s$:
\begin{equation}
	|\lambda-\lambda_{s}|^{2}=\rho_{Y}^{2}B_{Y}^{2},\label{squared_bias_lam}
\end{equation}
where $\rho_{Y}^{2}:=Cor^{2}(\alpha-\alpha_{s},g_{Y}-g_{Ys})$, $B_{Y}^{2}=C_{Y}^{2}C_{\alpha}^{2}S_{Y}^{2}$
and
\[C_{Y}^{2} = \frac{E[(g_Y-g_{Y_s})^2]}{E[(Y-g_{Y_s})^2]}, C_{\alpha}^{2}  =  \frac{E[(\alpha - \alpha_s)^2]}{E[\alpha_{s}^{2}]}, S_{Y}^{2} = E[(Y-g_{Ys})^{2}]E[\alpha_{s}^{2}]=\sigma_{Y_{s}}^{2}v_{s}^{2}.\]
$C_Y$ and $C_{\alpha}$ are referred to sensitivity parameters in the OVB analysis, and in practice, they can be specified by researchers according to domain knowledge of the empirical study. $S_{Y}^{2}$ can be directly estimated with data. With the above results, we can have 
\begin{equation}
	\lambda^{-}\leq\lambda\leq\lambda^{+},\label{OVB_bound_lambda}
\end{equation}
where $\lambda^{-}:=\lambda_{s}-|\rho_{Y}|B_{Y}$ and $\lambda^{+}:=\lambda_{s}+|\rho_{Y}|B_{Y}$
by using the fact that $C_Y$, $C_{\alpha}$ and $S_Y$ are all nonnegative. Similarly, for $\gamma$ and $\gamma_{s}$, by using the result in \citet{CCNSS_2024}, we can have:
\begin{align}
	\gamma-\gamma_{s} & =E[(\alpha-\alpha_{s})(g_{D}-g_{Ds})].\label{OVB_gamma}
\end{align} A key condition to achieving (\ref{OVB_gamma}) is that $E[g_{D_s}(\alpha-\alpha_s)]=0$, which holds for LATE, LATT and PLIVM. Then following similar procedures for deriving (\ref{squared_bias_lam}), we can have: 
\begin{equation}
	|\gamma-\gamma_{s}|^{2}=\rho_{D}^{2}B_{D}^{2},\label{squared_bias_gam}
\end{equation} where
$\rho_{D}^{2}:=Cor^{2}(\alpha-\alpha_{s},g_{D}-g_{Ds})$, $B_{D}^{2}=C_{D}^{2}C_{\alpha}^{2}S_{D}^{2}$
and
\[C_{D}^{2} = \frac{E[(g_D-g_{Ds})^2]}{E[(D-g_{Ds})^2]}, S_{D}^{2}= E[(D-g_{Ds})^{2}]E[\alpha_{s}^{2}] = \sigma_{Ds}^{2}v_{s}^{2}.\]
Finally, it can be shown that
\begin{equation}
	\gamma^{-}\leq\gamma\leq\gamma^{+},\label{OVB_bound_gamma}
\end{equation}
where $\gamma^{-}:=\gamma_{s}-|\rho_{D}|B_{D}$ and $\gamma^{+}:=\gamma_{s}+|\rho_{D}|B_{D}$
by using the fact that $C_D$, $C_{\alpha}$ and $S_D$ are all nonnegative.

\subsection{Constructing the OVB Bound for $\theta$}
Combining the above results yields the following partial identification result for $\phi_{t}$:
\begin{equation}
	\min\{ \lambda^{-}-\gamma^{+}t,\lambda^{-}-\gamma^{-}t\} \leq\phi_{t}\leq\max\{ \lambda^{+}-\gamma^{+}t,\lambda^{+}-\gamma^{-}t\} \label{OVB_phi}
\end{equation}for $t\in \boldsymbol{\Theta}_0$. With some algebra, (\ref{OVB_phi}) can be further rewritten as:
\begin{equation}
	\phi_{t}^{-}\leq\phi_{t}\leq \phi_{t}^{+},\label{OVB_AR}
\end{equation}
where
\begin{align*}
	\phi_{t}^{+} & =\lambda^{+}-\gamma^{-}t1\{t\geq0\}-\gamma^{+}t1\{t<0\},\\
	\phi_{t}^{-} & =\lambda^{-}-\gamma^{+}t1\{t\geq0\}-\gamma^{-}t1\{t<0\}.
\end{align*}
Using the result in (\ref{OVB_AR}), we derive the following partial identification results for $\theta_0$. 
\begin{theorem}
	Suppose that $\theta=\theta_0$ and $\phi_{\theta_0} = \lambda - \gamma \theta_0$ satisfies (\ref{OVB_AR}): 	$\phi_{\theta_0}^{-}\leq\phi_{\theta_0}\leq \phi_{\theta_0}^{+}$.
	\begin{enumerate}
		\item When $(\gamma^{-},\gamma^{+})\in\mathbb{R}^{++}$:
		\begin{itemize}
			\item If $(\lambda^{-},\lambda^{+})\in\mathbb{R}^{++}$, then
			$\theta_0\in [\lambda^{-}/\gamma^{+},\lambda^{+}/\gamma^{-}].$
			\item If $(\lambda^{-},\lambda^{+})\in\mathbb{R}^{--}$, then
			$\theta_0\in [\lambda^{-}/\gamma^{-},\lambda^{+}/\gamma^{+}].$
			\item If $\lambda^{-}$ and $\lambda^{+}$ have different signs, then
			$\theta_0\in [\lambda^{-}/\gamma^{-},\lambda^{+}/\gamma^{-}]$. 
		\end{itemize}
		\item When $(\gamma^{-},\gamma^{+})\in\mathbb{R}^{--}$:
		\begin{itemize}
			\item If $(\lambda^{-},\lambda^{+})\in\mathbb{R}^{++}$, then $\theta_0 \in [\lambda^{+}/\gamma^{+},\lambda^{-}/\gamma^{-}].$
			\item If $(\lambda^{-},\lambda^{+})\in\mathbb{R}^{--}$, then $\theta_0 \in[\lambda^{+}/\gamma^{-},\lambda^{-}/\gamma^{+}].$
			\item If $\lambda^{+}$ and $\lambda^{-}$ have different signs, then $\theta_0 \in [\lambda^{+}/\gamma^{+},\lambda^{-}/\gamma^{+}]$. 
		\end{itemize}
		\item When $\gamma^{-}\neq 0$ and $\gamma^{+}\neq 0$ and they have different signs:
		\begin{itemize}
			\item If $(\lambda^{-},\lambda^{+})\in\mathbb{R}^{++}$, then $\theta_{0}\in (-\infty,\lambda^{-}/\gamma^{-}]\cup[\lambda^{-}/\gamma^{+},\infty).$
			\item If $(\lambda^{-},\lambda^{+})\in\mathbb{R}^{--}$, then $\theta_{0}\in(-\infty,\lambda^{+}/\gamma^{+}]\cup[\lambda^{+}/\gamma^{-},\infty)$. \item If $\lambda^{+}$ and $\lambda^{-}$ have different signs, then $\theta_{0}\in\left(-\infty,\infty\right)$. 
		\end{itemize}
	\end{enumerate}
\end{theorem}

To prove Theorem 1, note that if $\theta_{0}$ is the true value of $\theta$, $0\in [\phi_{\theta_{0}}^{-},\phi_{\theta_{0}}^{+}]$. This requires that both $\phi_{\theta_{0}}^{+}\geq0$ and $\phi_{\theta_{0}}^{-}\leq0$ hold. Therefore:
\[
\theta_{0}\in\{t\in \boldsymbol{\Theta}_0:\{\phi_{t}^{+}\geq0\}\cap\{\phi_{t}^{-}\leq0\}\}.
\]
In addition, when $(\gamma^{-},\gamma^{+})\in\mathbb{R}^{++}$, if $(\lambda^{-},\lambda^{+})\in\mathbb{R}^{++}$, $\theta_{0}>0$. If $(\lambda^{-},\lambda^{+})\in\mathbb{R}^{--}$,
$\theta_{0}<0$. If $\lambda^{-}$ and $\lambda^{+}$ have different signs, $\theta_{0}\in[-c_1,c_2]$ where $(c_1,c_2)>0$ are some constants. Similar arguments can be applied to derive the OVB bounds of $\theta_0$ when $(\gamma^{-},\gamma^{+})\in\mathbb{R}^{--}$.

Let $(\hat{\lambda}^{-},\hat{\lambda}^{+},\hat{\gamma}^{-},\hat{\gamma}^{+})$ denote some estimates of $(\lambda^{-},\lambda^{+},\gamma^{-},\gamma^{+})$. In practice, we can apply the result of Theorem 1 to estimate the upper and lower bounds for $\theta_0$ using $(\hat{\lambda}^{-},\hat{\lambda}^{+},\hat{\gamma}^{-},\hat{\gamma}^{+})$. If the signs of $(\hat{\gamma}^{-},\hat{\gamma}^{+})$ are the same, the OVB bounds can be directly estimated following points 1. and 2. of Theorem 1. However, if $(\hat{\gamma}^{-},\hat{\gamma}^{+})$ have different signs, the situation becomes more complicated. According to point 3. of Theorem 1, when $(\gamma^{-},\gamma^{+})\neq (0,0)$ and have different signs, the partially identified set for $\theta_0$ is either (a) split into disjoint segments of the real line (when $(\lambda^{-},\lambda^{+})$ have the same sign); or (b) the entire real line (when $(\lambda^{-},\lambda^{+})$ have different signs). In case (a), zero is not included in the OVB bound; in (b), it is. These results, particularly case (a), are hard to be interpreted. Therefore for practically applying Theorem 1, we recommend first checking whether $(\hat{\gamma}^{-},\hat{\gamma}^{+})\neq (0,0)$ and have the same sign: $(\hat{\gamma}^{-},\hat{\gamma}^{+})\in\mathbb{R}^{++}$ or $(\hat{\gamma}^{-},\hat{\gamma}^{+})\in\mathbb{R}^{--}$. If this condition fails, we suggest stopping here and reporting that the first-stage estimation fails when the OVB is concerned. 
If the condition holds, we proceed with results of points 1. and 2. of Theorem 1 to construct the OVB bound for $\theta_0$.

\subsection{Sensitivity Parameters in the OVB Analysis}
The sensitivity parameters $C_{\alpha}$, $C_{Y}$ and $C_D$ play crucial roles in the OVB analysis. In this section, we elaborate on their properties and show that serve as measures of the strength of the omitted variable. 

Let $R^{2}_{U_1\sim U_2}:=\text{Var}(U_2)/\text{Var}(U_1)$ denote the ratio of variances between two random variables $U_2$ and $U_1$. If $U_2=E[U_1|U_3]$ holds, 
\begin{equation}
	R^{2}_{U_1\sim U_2}=R^{2}_{U_1\sim E[U_1|U_3]}=\text{Var}(E[U_1|U_3])/\text{Var}(U_1):=\eta_{U_1\sim U_3}^{2},\label{R2_U1U2}
\end{equation} where $\eta_{U_1\sim U_3}^{2}$ denotes the non-parametric R-squared (Pearson's correlation ratio) between $U_1$ and $U_3$. 

In the cases of PLIVM, LATE and LATT, we have $E[\alpha]=E[\alpha_{s}]=0$, and therefore $\text{Var}(\alpha)=E[\alpha^{2}]$ and $\text{Var}(\alpha_{s})=E[\alpha_{s}^{2}]$. Then using the fact that $E[\alpha_{s}(\alpha-\alpha_{s})]=0$, we also can have
$E[(\alpha-\alpha_{s})^{2}]=E[\alpha^{2}]-E[\alpha_{s}^{2}]\geq0$ and express $C_{\alpha}^{2}$ as:
\begin{equation}
	C_{\alpha}^{2} = \frac{1-R^{2}_{\alpha\sim \alpha_s}}{R^{2}_{\alpha\sim \alpha_s}}.
	\label{C_alpha}
\end{equation}For LATE and LATT, since $E[\alpha|W_s] = \alpha_s$, we have 
$R_{\alpha\sim\alpha_{s}}^{2}=\eta_{\alpha\sim W_{s}}^{2}$, the
nonparametric $R^{2}$ between $\alpha$ and $W_{s}$. 
However, $R_{\alpha\sim\alpha_{s}}^{2}=\eta_{\alpha\sim W_{s}}^{2}$ does not hold for PLIVM, since $E[\alpha|W_s] \neq \alpha_s$ for this case. 

For LATE, using $\text{Var}(Z|X,A)=\pi(X,A)(1-\pi(X,A))$, we obtain $E[\alpha^{2}]=E\left[1/\text{Var}(Z|X,A)\right]$, which is the expected
precision of prediction on $Z$ using $(X,A)$. Similarly, $E[\alpha_{s}^{2}]=E\left[1/\text{Var}(Z|X)\right].$
Therefore 
\[
1-R^2_{\alpha\sim\alpha_s} =\frac{E\left[\frac{1}{\text{Var}(Z|X,A)}\right]-E\left[\frac{1}{\text{Var}(Z|X)}\right]}{E\left[\frac{1}{\text{Var}(Z|X,A)}\right]},
\]which quantifies the (absolute) decrease in expected prediction precision for $Z$ when $A$ is omitted in the model. Note that $1-R^2_{\alpha\sim\alpha_s}$ is bounded between 0 and 1. Using the result in (\ref{C_alpha}), we also have:
\begin{align*}
	C_{\alpha}^{2} & =\frac{E\left[\frac{1}{\text{Var}(Z|X,A)}\right]-E\left[\frac{1}{\text{Var}(Z|X)}\right]}{E\left[\frac{1}{\text{Var}(Z|X)}\right]},
\end{align*}
which represents the additional gain in expected prediction precision for $Z$ when $(X,A)$ are included into the model, relative to using $X$ alone. For PLIVM, $\alpha(W)=(Z-E[Z|A,X])/E[(Z-E[Z|A,X])^{2}]$
and $\alpha_{s}(W_{s})=(Z-E[Z|X])/E[(Z-E[Z|X])^{2}]$. Then it can be shown that $R_{\alpha\sim\alpha_{s}}^{2}=E[(Z-E[Z|A,X])^{2}]/E[(Z-E[Z|X])^{2}]$ and 
\begin{align*}
	C_{\alpha}^{2} & =\frac{E[(Z-E[Z|X])^{2}]-E[(Z-E[Z|A,X])^{2}]}{E[(Z-E[Z|X,A])^{2}]}\\
	& =\frac{\eta_{Z\sim A|X}^{2}}{1-\eta_{Z\sim A|X}^{2}},
\end{align*}
which captures an increase in the MSE for predicting $Z$
when $A$ is absent in the model, relative to that for predicting $Z$ when both $(Z,A)$ present. In the second equality, $\eta_{Z\sim A|X}^{2} = 1-R^2_{\alpha\sim\alpha_s}$ 
is the nonparametric partial $R^{2}$ between $Z$ with $A$, conditional
on $X$, which is defined as
\[
\eta_{Z\sim A|X}^{2}:=\frac{E[\text{Var}(Z|X)]-E[\text{Var}(Z|A,X)]}{E[\text{Var}(Z|X)]}=\frac{\eta_{Z\sim A,X}^{2}-\eta_{Z\sim X}^{2}}{1-\eta_{Z\sim X}^{2}}.
\]
$\eta_{Z\sim A|X}^{2}$ captures the extra explanatory power that $A$ provides for $Z$, beyond what is already explained by $X$, relative
to $1-\eta_{Z\sim X}^{2}$, the remaining unexplained variation of $Z$ after conditioning on $X$. 

For $C_{Y}^{2}$ (or $C_{D}^{2}$), with the notations in (\ref{R2_U1U2}), we can express $C_{Y}^{2}$ as:
\[C_{Y}^{2} = R^{2}_{Y-g_{Ys} \sim g_Y-g_{Ys}}
\]using the identities $E[g_{Y}g_{Y_{s}}]=E[Yg_{Y_{s}}]=E[g_{Y_{s}}^{2}]$. Furthermore:
\begin{eqnarray*}
	C_{Y}^{2} & = &\frac{E[(g_{Y}-g_{Ys})^{2}]}{E[(Y-g_{Ys})^{2}]}\\
	& = &\frac{E[\text{Var}(Y|W_{s})]-E[\text{Var}(Y|W)]}{E[\text{Var}(Y|W_{s})]}\\
	& = &\eta_{Y\sim A|Z,X}^{2},
\end{eqnarray*}
which is the nonparametric partial $R^{2}$ between $Y$ with $A$, conditional
on $(Z,X)$. The quantity $\eta_{Y\sim A|Z,X}^{2}$ reflects the additional explanatory power of $A$ beyond what $(Z,X)$ already provides,
relative to $1-\eta_{Y\sim Z,X}^{2}$, the unexplained variation in $Y$ given $(Z,X)$. Alternatively, $C_{Y}^{2}$ can also be interpreted as the proportional reduction in the MSE for predicting $Y$ when $A$ is additionally included into the model with $(Z,X)$, comparing to using $(Z,X)$ alone. 
Thus the sensitivity parameters $C_Y^{2}$, $C_D^{2}$ and $C_{\alpha}^{2}$ measure the gains in predictive accuracy for $Y$, $D$ and $Z$, respectively, when variable $A$ is included in the models given $X$.

To compute the OVB bounds in (\ref{OVB_bound_lambda}), (\ref{OVB_bound_gamma}) and (\ref{OVB_AR}), we need to assign values to the sensitivity parameters $C_Y$, $C_{\alpha}$ and $C_D$. As discussed above, these parameters capture how the  omitted variable $A$ affects the weight $\alpha$ and predictions for $Y$ and $D$, when only $(Z,X)$ are used. Therefore setting their values is equivalent to assessing the importance of the omitted variable $A$ in determining $\alpha$ and predicting $Y$ and $D$. Since the parameters of interest are constituted by the weight $\alpha$ and predictions on $Y$ and $D$, carefully selecting the values of the sensitivity parameters is crucial for quantifying the bias due to the omission of $A$. This can be done by leveraging the researcher’s domain knowledge, by estimating them from data through benchmarking analysis (see the discussion below), or by combining both approaches.

Finally, the sensitivity parameters $C_{\alpha}^{2}$, $C_Y^{2}$ and $C_D^{2}$ are all unit-free, as they are scaled by factors that eliminate dependence on measurement units. This scale-invariance ensures that the parameters are comparable across different variables and empirical contexts, regardless of their units of measurement. Since these parameters are derived from variance ratios (e.g., nonparametric R-squared), they quantify proportional improvements in predictive accuracy rather than absolute changes. This property facilitates meaningful interpretation, robust sensitivity analysis, and consistent calibration of parameter values—particularly when conducting simulations or assessing the importance of omitted variables whose scales may be unknown or heterogeneous. As a result, the unit-free nature of these sensitivity measures enhances both the generalizability and practical relevance of the OVB analysis.


\subsection{Benchmarking Analysis}
Following \citet{CH_2022} and \citet{CCNSS_2024}, we conduct a benchmarking analysis by imposing a requirement that the explanatory power gained from including the omitted variable $A$ should be comparable to that obtained from specific observable variables. The primary objective of this analysis is to establish reasonable bounds for restricting the maximum values of the sensitivity parameters $C_{\alpha}^{2}$, $C_Y^{2}$ and $C_D^{2}$. To achieve this, we first show that $C_{\alpha}^{2}$, $C_{Y}^{2}$ and $C_{D}^{2}$ in the OVB bounds can be expressed as functions of the relative strength of the omitted variable $A$ comparable to other observable variables. Let $X_{-j}$ denote the set of all other observable variables when $X_{j}$ is excluded from $X$. Let $W_{-j}=(Z,X_{-j},A)$, $W_{s,-j}=(Z,X_{-j})$ and: 
\begin{align*}
	\alpha_{s,-j}(W_{s,-j}) & :=\frac{Z}{\pi(X_{-j})}-\frac{(1-Z)}{1-\pi(X_{-j})},\\
	g_{Ys,-j}(W_{s,-j}) & :=E[Y|W_{s,-j}]=E[Y|Z,X_{-j}],\\
	g_{Ds,-j}(W_{s,-j}) & :=E[D|W_{s,-j}]=E[D|Z,X_{-j}].
\end{align*}
As before, we will use the abbreviations $\alpha_{s,-j}$, $g_{Ys,-j}$ and $g_{Ds,-j}$
to denote $\alpha_{s,-j}(W_{s,-j})$, $g_{Ys,-j}(W_{s,-j})$ and $g_{Ds,-j}(W_{s,-j})$. For
$\alpha_{s}$, define the gain in explanatory power from including $X_{j}$, given $X_{-j}$ as:
\[
1-R_{\alpha_{s}\sim\alpha_{s,-j}}^{2}=1-\frac{E[\alpha_{s,-j}^{2}]}{E[\alpha_{s}^{2}]}.
\]
The gain in explanatory power from including $A$, $1-R_{\alpha\sim\alpha_{s}}^{2}$ can then be expressed as:
\begin{eqnarray}
	1-R_{\alpha\sim\alpha_{s}}^{2} 
	&=&1-\frac{E[\alpha_{s}^{2}]}{E[\alpha_{s,-j}^{2}]}\frac{E[\alpha_{s,-j}^{2}]}{E[\alpha^{2}]}\nonumber \\
	&=&\frac{(1-R_{\alpha\sim\alpha_{s,-j}}^{2})-(1-R_{\alpha_{s}\sim\alpha_{s,-j}}^{2})}{R_{\alpha_{s}\sim\alpha_{s,-j}}^{2}}.\label{C2_alpha1}
\end{eqnarray}
Note that $1-R_{\alpha\sim\alpha_{s,-j}}^{2}$ measures the gain in explanatory power from including $(A,X_{j})$, given $X_{-j}$. The numerator in (\ref{C2_alpha1}) therefore captures the extra explanatory power from including $A$ beyond that provided by $X_{j}$, given $X_{-j}$. Define the relative strength of $A$ to $X_{j}$ for $\alpha$ as: 
\begin{equation}
	k_{\alpha}=\frac{R_{\alpha_{s}\sim\alpha_{s,-j}}^{2}-R_{\alpha\sim\alpha_{s,-j}}^{2}}{1-R_{\alpha_{s}\sim\alpha_{s,-j}}^{2}},\label{k_alpha}
\end{equation}
which is a ratio of the extra gain in explanatory power from including $(X_j,A)$, compared to that from including $X_{j}$ only, given $X_{-j}$. Therefore the quantity $k_{\alpha}$ measures relative importance of $A$ compared to $X_{j}$ in explaining $\alpha$, given $X_{-j}$. A value $k_{\alpha}\leq1$ indicates that $A$ is relatively less important than $X_{j}$ for explaining $\alpha$, given $X_{-j}$. Furthermore, it can be shown that 
\[
1-R_{\alpha\sim\alpha_{s}}^{2}=k_{\alpha}G_{\alpha},
\]
where 
\begin{equation*}
	G_{\alpha}=\frac{1-R_{\alpha_{s}\sim\alpha_{s,-j}}^{2}}{R_{\alpha_{s}\sim\alpha_{s,-j}}^{2}}
	\label{G_alpha}
\end{equation*}
is a quantity that can be estimated. This leads directly to:
\begin{equation*}
	C_{\alpha}^{2}=\frac{1-R_{\alpha\sim\alpha_{s}}^{2}}{R_{\alpha\sim\alpha_{s}}^{2}}=\frac{k_{\alpha}G_{\alpha}}{1-k_{\alpha}G_{\alpha}}.
\end{equation*}

For $C_{Y}^{2}$, at first note that\footnote{The same result and derivation are applied to $C_{D}^{2}$.}
\begin{eqnarray}
	C_{Y}^{2} 
	&=&\frac{E[Y^{2}]-E[g_{Ys}^{2}]-(E[Y^{2}]-E[g_{Y}^{2}])}{E[(Y-g_{Ys})^{2}]}\nonumber \\
	&=&\frac{E[(Y-g_{Ys})^{2}]-E[(Y-g_{Y})^{2}]}{E[(Y-g_{Ys})^{2}]}\nonumber \\
	& = &1-R_{Y-g_{Ys}\sim Y-g_{Y}}^{2}.\label{C2_Y1}
\end{eqnarray}by using the result $E(Y-g_Y)^{2}=E[Y^{2}]-E[g_{Y}^{2}]$.
Equation (\ref{C2_Y1}) captures the relative reduction in the mean squared error (MSE) in predicting $Y$ when including the omitted variable $A$. Following similar argument for deriving expression of $1-R_{\alpha\sim\alpha_{s}}^{2}$, we also have:
\begin{eqnarray}
	1-R_{Y-g_{Ys}\sim Y-g_{Y}}^{2} 
	&=&1-\frac{E[(Y-g_{Ys,-j})^{2}]}{E[(Y-g_{Y})^{2}]}\frac{E[(Y-g_{Ys})^{2}]}{E[(Y-g_{Ys,-j})^{2}]}\nonumber\\
	& =&\frac{(1-R_{Y-g_{Ys,-j}\sim Y-g_{Y}}^{2})-(1-R_{Y-g_{Ys,-j}\sim Y-g_{Ys}}^{2})}{R_{Y-g_{Ys,-j}\sim Y-g_{Ys}}^{2}}.\label{C2_Y2}
\end{eqnarray}
The numerator in (\ref{C2_Y2}) represents the additional reduction in MSE in predicting $Y$ from including $(X_j,A)$, compared to that from including $X_{j}$ only, given $X_{-j}$. Define the relative strength of $A$ to $X_{j}$ for predicting $Y$ as
\begin{equation*}
	k_{Y}=\frac{R_{Y-g_{Ys,-j}\sim Y-g_{Ys}}^{2}-R_{Y-g_{Ys,-j}\sim Y-g_{Y}}^{2}}{1-R_{Y-g_{Ys,-j}\sim Y-g_{Ys}}^{2}}.
	\label{k_Y}
\end{equation*}
This yields
\begin{equation*}
	C_Y^{2} =1-R_{Y-g_{Ys}\sim Y-g_Y}^{2} =k_{Y}G_{Y},
\end{equation*}
where 
\begin{equation*}
	G_{Y}=\frac{1-R_{Y-g_{Ys,-j}\sim Y-g_{Ys}}^{2}}{R_{Y-g_{Ys,-j}\sim Y-g_{Ys}}^{2}}.\label{G_Y}
\end{equation*}
is a quantity that can be estimated. 

Estimating $G_{\alpha}$ and $G_Y$ involves with estimating the variance ratios $R_{\alpha_s\sim\alpha_{s,-j}}^{2}$ and $R_{Y-g_{Ys,-j}\sim Y-g_{Ys}}^{2}$, which can be directly computed from available data. 
In addition, for estimating $\alpha_{s,-j}$ and $g_{Ys,-j}$, there is no restriction on the number of excluded variables $X_j$. That is, we may exclude a group of variables simultaneously if it is necessary.\footnote{For example, variables such as age are often represented categorically. In practice, we may exclude all age-related variables when estimating $\alpha_{s,-j}$ and $g_{Ys,-j}$. This approach is equivalent to treating $X_j$ as a vector that includes these age variables.} 
This flexibility facilitates a richer analysis on robustness of parameter estimations to the omitted variable bias.

\section{Estimation and Inference for the OVB Bound}

To estimate the OVB bound, we need to estimate $(\lambda_s, \gamma_s)$, the short version of $(\lambda,\gamma)$, as well as $(v^{2}_s,\sigma^{2}_{Y_s},\sigma^{2}_{D_s})$, along with the calibrated values of the sensitivity parameters $C_{\alpha}$, $C_Y$ and $C_D$ and correlation coefficients $|\rho_{Y}|$ and $|\rho_D|$. 
In our empirical application, we employ the double machine learning (DML) estimators combined with the median method \citep{CCDDHNR_2018} to estimate  $(\lambda_s, \gamma_s, v^{2}_s,\sigma^{2}_{Y_s},\sigma^{2}_{D_s})$. The DML estimator integrates an estimator satisfying the Neyman orthogonality with K-fold cross fitting. We begin by introducing the former, which can be derived using the influence functions (IFs). 

We first consider estimating the OVB bound for LATE. In this case, the IFs of $\lambda_{s}$ and $\gamma_s$ are given by: 
\begin{equation}
	\psi_{\lambda_{s}}(Y,W_s) =\bar{\psi}(Y,W_s)-\lambda_{s},\text{ }\psi_{\gamma_s}(D,W_s) =\bar{\psi}(D,W_s)-\gamma_s.\label{IF_LATE_s}
\end{equation}where
\begin{eqnarray*}
	\bar{\psi}(Y,W_s) &=&  \frac{Z}{\pi_{s}(X)}(Y-E[Y|Z=1,X])-\frac{1-Z}{1-\pi_{s}(X)}(Y-E[Y|Z=0,X])+\nonumber\\
	&&E[Y|Z=1,X]-E[Y|Z=0,X],\\
	\bar{\psi}(D,W_s) &=& \frac{Z}{\pi_{s}(X)}(D-E[D|Z=1,X])-\frac{1-Z}{1-\pi_{s}(X)}(D-E[D|Z=0,X])+\nonumber\\
	&&E[D|Z=1,X]-E[D|Z=0,X].
\end{eqnarray*}By using the moment conditions $E[\psi_{\lambda_{s}}(Y,W_s)]=E[\psi_{\gamma_{s}}(D,W_s)]=0$, we identify: 
\begin{equation}
	\lambda_{s} =E[\bar{\psi}(Y,W_s)],\text{ }\gamma_{s} =E[\bar{\psi}(D,W_s)].\label{LATE_DML}
\end{equation}
It can be shown that the estimators based on the IFs (\ref{IF_LATE_s}) satisfy Neyman orthogonality. Given that $\gamma_s\neq0$, the short version of $\theta$ is $\theta_{s}=\lambda_{s}/\gamma_{s}$,
which can also be identified by solving the moment condition $E[\psi_{\theta_{s}}(Y,D,W_s)]=0$, where 
\[
\psi_{\theta_{s}}(Y,D,W_s)= \bar{\psi}(Y,W_s)-\bar{\psi}(D,W_s)\theta_s
\]

For the case of LATT, the IFs for $\lambda_{s}$ and $\gamma_s$ are given by \citep{Hahn_1998}:
\begin{equation}
	\psi_{\lambda_{s}}(Y,W_s)= \tilde{\psi}(Y,W_s)-\frac{Z\lambda_{s}}{P_{Z}},\text{ }
	\psi_{\gamma_{s}}(D,W_s)= \tilde{\psi}(D,W_s)-\frac{Z\gamma_{s}}{P_{Z}}\label{LATT_IF_s}
\end{equation}where
\begin{eqnarray*}
	\tilde{\psi}(Y,W_s)&= & \frac{Z}{P_{Z}}(Y-E[Y|Z=1,X])-\frac{1-Z}{P_{Z}}\frac{\pi_s(X)}{1-\pi_s(X)}(Y-E[Y|Z=0,X])+\\
	& &\frac{Z}{P_{Z}}(E[Y|Z=1,X]-E[Y|Z=0,X]),\\
	\tilde{\psi}(D,W_s)& = & \frac{Z}{P_{Z}}(D-E[D|Z=1,X])-\frac{1-Z}{P_{Z}}\frac{\pi_s(X)}{1-\pi_s(X)}(D-E[D|Z=0,X])+\\
	& &\frac{Z}{P_{Z}}(E[D|Z=1,X]-E[D|Z=0,X]).
\end{eqnarray*}Again, by setting $E[\psi_{\lambda_{s}}(Y,W_s)]=E[\psi_{\gamma_{s}}(D,W_s)]=0$, $\lambda_{s}$ and $\gamma_{s}$ can be identified as:
\begin{equation}
	\lambda_{s} =E[\tilde{\psi}(Y,W_s)],\text{ }
	\gamma_{s} =E[\tilde{\psi}(D,W_s)].\label{LATT_DML}
\end{equation}
The estimators also 
satisfy Neyman orthogonality. Given that $\gamma_s\neq0$, the short version of $\theta$ is $\theta_{s}=\lambda_{s}/\gamma_{s}$,
which can also be solved by using the moment condition $E[\tilde{\psi}_{\theta_{s}}]=0$,
where 
\[
\psi_{\theta_{s}}(Y,D,W_s)= \tilde{\psi}(Y,W_s)-\tilde{\psi}(D,W_s)\theta_s.
\]

For PLIVM in (\ref{Y_plm}) and (\ref{D_plm}), we use the following IFs of $\lambda_{s}$ and $\gamma_s$ (Robinson style score functions) to estimate $\lambda_{s}$ and $\gamma_{s}$:
\begin{eqnarray}
	\psi_{\lambda_{s}}(Y,W_s) 
	& = &(Y-m(X))(Z-l(X))-\lambda_{s}(Z-l(X))^{2}.\label{IF_PLIVM_lambda_s}\\
	\psi_{\gamma_{s}}(D,W_s) 
	& = &(D-r(X))(Z-l(X))-\gamma_{s}(Z-l(X))^{2}\label{IF_PLIVM_gamma_s},
\end{eqnarray}
where $m(X):=E[Y|X]$, $r(X):=E[D|X]$ and $l(X):=E[Z|X]$. Setting $E[\psi_{\lambda_{s}}(Y,W_s)]=E[\psi_{\gamma_{s}}(D,W_s)]=0$ yields: 
\begin{equation}
	\lambda_{s} = \frac{E[(Y-m(X))(Z-l(X))]}{E[(Z-l(X))^{2}]},\text{ }
	\gamma_{s} = \frac{E[(D-r(X))(Z-l(X))]}{E[(Z-l(X))^{2}]}.\label{PLIVM_DML}
\end{equation}
Given that $\gamma_s\neq 0$, the short version of $\theta$ is $\theta_{s}=\lambda_{s}/\gamma_{s}$,
which can also be obtained with the moment condition $E[\psi_{\theta_{s}}(Y,D,W_s)]=0$,
where 
\[
\psi_{\theta_{s}}(Y,D,W_s)=[Y-m(X)-\theta_{s}(D-r(X))][Z-l(X)]
\]
is a Robinson style score function for $\theta_{s}$ with $Z$ as the instrumental variable. 

The sample analogues of (\ref{LATE_DML}), (\ref{LATT_DML}) and (\ref{PLIVM_DML}) are used as estimators for estimating $\lambda_s$ and $\gamma_s$ in conjunction with K-fold cross-fitting (see Section 2.5.2). Inside these estimators, nuisance parameters, such as $\pi_{s}(X)$, $E[Y|Z=1,D]$ and $E[D|Z=1,X]$ etc., can be estimated using appropriate parametric or nonparametric models, potentially enhanced with various machine learning methods (e.g., random forest or the lasso), especially when the dimension of $X$ is large and/or their functional forms are complex. In our empirical application, we use \textit{random forest}\footnote{The random forest is conducted with function \texttt{ranger} in \texttt{R} package \texttt{ranger}.} with K-fold cross-fitting to estimate these nuisance parameters. The estimate of the short version $\theta_s$ is crucial for the purposes of comparison and statistical inference in the empirical analysis. However, for estimating the OVB bounds for $\theta$ shown in Theorem 1, we only need estimates of $(\lambda_s, \gamma_s,v^{2}_s,\sigma^{2}_{Y_s},\sigma^{2}_{D_s})$, and estimating $\theta_s$ is not required.

\subsection{Confidence Interval for the OVB Bound}
We now turn to construction of confidence interval (C.I) for the OVB bound. The C.I. may serve as a measure to determine whether the statistical significance of the initial estimate persists after accounting for OVB. 
Let $\zeta_{Y,\alpha}=|\rho_{Y}|C_{Y}C_{\alpha}$ and $\zeta_{D,\alpha}=|\rho_{D}|C_{D}C_{\alpha}$. The influence functions (IFs) of $\lambda^{+}$ and $\lambda^{-}$ are given by \citep{CCNSS_2024}:
\begin{equation*}
	\psi_{\lambda^{+}}=\psi_{\lambda_{s}}+\zeta_{Y,\alpha}\psi_{S_Y},\text{ }
	\psi_{\lambda^{-}}=\psi_{\lambda_{s}}-\zeta_{Y,\alpha}\psi_{S_Y},
\end{equation*}where
\[
\psi_{S_{Y}}=\frac{\sigma_{Y_{s}}^{2}\psi_{v_{s}^{2}}+v_{s}^{2}\psi_{\sigma_{Y_{s}}^{2}}}{2S_{Y}},
\] is the IF of $S_Y$ and 
\[\psi_{\sigma_{Ys}^{2}}= (Y-g_{Ys})^{2}-\sigma_{Ys}^{2}, \text{ } \psi_{v_{s}^{2}}=\alpha_{s}^{2}-v_{s}^{2}\] 
are IFs of $\sigma_{Ys}^{2}=E[(Y-g_{Ys})^{2}]$ and $v_{s}^{2}=E[\alpha_{s}^{2}]$. If the DML estimators for estimating $\lambda^{+}$
and $\lambda^{-}$, denoted by $\hat{\lambda}^{+}$
and $\hat{\lambda}^{-}$, satisfy certain regularity conditions \citep{CCDDHNR_2018}, then $\hat{\lambda}^{+}$ and $\hat{\lambda}^{-}$ exhibit asymptotic normality: 
\[
\sqrt{n}(\hat{\lambda}^{+}-\lambda^{+}) \overset{a.}{\rightarrow} N\left(0,E\left[\psi_{\lambda^{+}}^{2}\right]\right),
\sqrt{n}(\hat{\lambda}^{-}-\lambda^{-}) \overset{a.}{\rightarrow} N\left(0,E\left[\psi_{\lambda^{-}}^{2}\right]\right).
\]Furthermore, the following one-sided covering properties hold:
\[
\lim_{n\rightarrow\infty}P(\lambda^{+}\leq\hat{\lambda}^{+}_{1-\tau})\geq 1-\tau, \lim_{n\rightarrow\infty}P(\lambda^{-}\geq \hat{\lambda}^{-}_{\tau})\geq 1-\tau,
\]where 
\begin{equation}
	\hat{\lambda}^{+}_{1-\tau} := \hat{\lambda}^{+}+\text{se}(\hat{\lambda}^{+})\Phi^{-1}(1-{\tau}), \hat{\lambda}^{-}_{\tau} := \hat{\lambda}^{-}-\text{se}(\hat{\lambda}^{-})\Phi^{-1}(1-{\tau}),\label{lambda_ovb_CI}
\end{equation}and $\text{se}(\hat{\lambda}^{-}):=\sqrt{\widehat{\text{Var}}(\hat{\lambda}^{-})/n}$ and $\text{se}(\hat{\lambda}^{+}):=\sqrt{\widehat{\text{Var}}(\hat{\lambda}^{+})/n}$ are the standard errors of $\hat{\lambda}^{+}$ and $\hat{\lambda}^{-}$, and $\Phi^{-1}(1-\tau)$ denotes the $(1-\tau)$-th quantile of the standard normal distribution.\footnote{In the following, we assume that $\tau\leq 0.5$.}

For $(\gamma^{+}, \gamma^{-})$, we also have similar results. The IFs of $\gamma^{+}$ and $\gamma^{-}$ are given by: 
\[
\psi_{\gamma^{+}}  =\psi_{\gamma_{s}}+\zeta_{D,\alpha}\psi_{S_D},\text{ }
\psi_{\gamma^{-}}  =\psi_{\gamma_{s}}-\zeta_{D,\alpha}\psi_{S_D},
\]where 
\[
\psi_{S_{D}} =\frac{\sigma_{D_{s}}^{2}\psi_{v_{s}^{2}}+v_{s}^{2}\psi_{\sigma_{D_{s}}^{2}}}{2S_{D}},
\]is the IF of $S_D$ and \[\psi_{\sigma_{Ds}^{2}}=(D-g_{Ds})^{2}-\sigma_{Ds}^{2}\]
is the IF of $\sigma_{Ds}^{2}=E[(D-g_{Ds})^{2}]$. Again, if the DML estimators for estimating $\gamma^{+}$ and $\gamma^{-}$, denoted by $\hat{\gamma}^{+}$ and $\hat{\gamma}^{-}$, satisfy certain regularity conditions, the asymptotic normality of $\hat{\gamma}^{+}$ and $\hat{\gamma}^{-}$ holds:
\[
\sqrt{n}(\hat{\gamma}^{+}-\gamma^{+}) \overset{a.}{\rightarrow} N\left(0,E\left[\psi_{\gamma^{+}}^{2}\right]\right),
\sqrt{n}(\hat{\gamma}^{-}-\gamma^{-}) \overset{a.}{\rightarrow} N\left(0,E\left[\psi_{\gamma^{-}}^{2}\right]\right).
\]Accordingly the following one-sided covering properties also hold: 
\[
\lim_{n\rightarrow\infty}P(\gamma^{+}\leq\hat{\gamma}^{+}_{1-\tau})\geq 1-\tau, \lim_{n\rightarrow\infty}P(\gamma^{-}\geq \hat{\gamma}^{-}_{\tau})\geq 1-\tau,
\]
where 
\begin{equation}
	\hat{\gamma}^{+}_{1-\tau}:= \hat{\gamma}^{+}+\text{se}(\hat{\gamma}^{+})\Phi^{-1}(1-{\tau}),
	\hat{\gamma}^{-}_{\tau}:=\hat{\gamma}^{-}-\text{se}(\hat{\gamma}^{-})\Phi^{-1}(1-{\tau}),\label{gamma_ovb_CI}
\end{equation}and $\text{se}(\hat{\gamma}^{-}):=\sqrt{\widehat{\text{Var}}(\hat{\gamma}^{-})/n}$ and $\text{se}(\hat{\gamma}^{+}):=\sqrt{\widehat{\text{Var}}(\hat{\gamma}^{+})/n}$ are the standard errors of $\hat{\gamma}^{+}$ and $\hat{\gamma}^{-}$.

We now show the asymptotic results of $\hat{\phi}_{t}^{+}$ and $\hat{\phi}_{t}^{-}$, the plug-in estimators constructed using $(\hat{\lambda}^{+},\hat{\gamma}^{+},\hat{\gamma}^{-})$ and $(\hat{\lambda}^{-},\hat{\gamma}^{+},\hat{\gamma}^{-})$ for estimating $\phi_{t}^{+}$ and $\phi_{t}^{-}$ in (\ref{OVB_AR}), under the assumptions 
that the parameters $(t,\rho_{Y},\rho_{D},C_{\alpha},C_{Y},C_{\alpha})$
are all fixed and some regularity conditions for the DML estimators hold. The derivations of the IFs of $\phi_{t}^{+}$ and $\phi_{t}^{-}$ and approximate variances of $\hat{\phi}_{t}^{+}$ and $\hat{\phi}_{t}^{-}$ relies using on the fact that $\phi_{t}^{+}$ and $\phi_{t}^{-}$ are linear functions of $(\lambda^{+}, \gamma^{+}, \gamma^{-})$ and $(\lambda^{-}, \gamma^{+}, \gamma^{-})$. 
\begin{theorem}
	Assume $(t,\rho_{Y},\rho_{D},C_{\alpha},C_{Y},C_{D})$ are all
	fixed. The influence functions of $\phi_{t}^{+}$ and $\phi_{t}^{-}$
	are given by:
	\[
	\psi_{\phi_{t}^{+}}=\mathbf{C}_{t}^{+\top}\boldsymbol{\psi},\psi_{\phi_{t}^{-}}=\mathbf{C}_{t}^{-\top}\boldsymbol{\psi},
	\]
	where 
	\begin{eqnarray*}
		\mathbf{C}_{t}^{+} & = & \left[1,-t,\left(\frac{\zeta_{Y,\alpha}\sigma_{Y_{s}}}{2v_{s}}+\frac{\zeta_{D,\alpha}\sigma_{D_{s}}|t|}{2v_{s}}\right),\frac{\zeta_{Y,\alpha}v_{s}}{2\sigma_{Y_{s}}},\frac{\zeta_{D,\alpha}v_{s}|t|}{2\sigma_{D_{s}}}\right]^{\top},\\
		\mathbf{C}_{t}^{-} & = & \left[1,-t,-\left(\frac{\zeta_{Y,\alpha}\sigma_{Y_{s}}}{2v_{s}}+\frac{\zeta_{D,\alpha}\sigma_{D_{s}}|t|}{2v_{s}}\right),-\frac{\zeta_{Y,\alpha}v_{s}}{2\sigma_{Y_{s}}},-\frac{\zeta_{D,\alpha}v_{s}|t|}{2\sigma_{D_{s}}}\right]^{\top},\\
		\boldsymbol{\psi} & = & \left[\psi_{\lambda s},\psi_{\gamma_{s}},\psi_{v_{s}^{2}},\psi_{\sigma_{Y_{s}}^{2}},\psi_{\sigma_{D_{s}}^{2}}\right]^{\top}.
	\end{eqnarray*} Suppose $(\lambda_{s},\gamma_{s},v_{s}^{2},\sigma_{Y_{s}}^{2},\sigma_{D_{s}}^{2})$
	are all estimated with the DML estimators. If Assumption 4.2 (for PLIVM) or 5.2 (for LATE) in \citet{CCDDHNR_2018}
	holds, then 
	\[
	\sqrt{n}(\hat{\phi}_{t}^{+}-\phi_{t}^{+})\overset{a.}{\rightarrow}N\left(0,\mathbf{C}_{t}^{+\top}\boldsymbol{\Omega}\mathbf{C}_{t}^{+}\right),\sqrt{n}(\hat{\phi}_{t}^{-}-\phi_{t}^{-})\overset{a.}{\rightarrow}N\left(0,\mathbf{C}_{t}^{-\top}\boldsymbol{\Omega}\mathbf{C}_{t}^{-}\right),
	\]
	where $\boldsymbol{\Omega}=\mathbf{J}_{0}^{-1}E[\boldsymbol{\psi}\boldsymbol{\psi}^{\top}]\mathbf{J}_{0}^{-1}$is
	the approximate covariance matrix of these DML estimators and $\mathbf{J}_{0}$
	is the Jacobian matrix.
\end{theorem}

In the case of PLIVM, $\mathbf{J}_0$ is a $5\times5$ diagonal matrix with diagonal elements: \[\left(-E\left[(Z-l(X))^{2}\right],-E\left[(Z-l(X))^{2}\right],-1,-1,-1\right).\]
For LATE and LATT, $\mathbf{J}_{0}$ is a negative identity matrix,
and the approximate variances of $\sqrt{n}(\hat{\phi}_{t}^{+}-\phi_{t}^{+})$
and $\sqrt{n}(\hat{\phi}_{t}^{-}-\phi_{t}^{-})$ can be simplified to $E[\psi_{\phi_{t}^{+}}^{2}]$
and $E[\psi_{\phi_{t}^{-}}^{2}]$. 

We next show that the $(1-\tau)$ OVB-adjusted C.I. for $\theta$ can be constructed using the results in Theorem 2. 
\begin{theorem}
	Let $\hat{\phi}_{t,1-\tau}^{+}$ denote the upper bound of $(1-\tau)$ C.I. of $\phi_{t}^{+}$, and $\hat{\phi}_{t,\tau}^{-}$ denote the lower bound of $(1-\tau)$ C.I. of $\phi_{t}^{-}$, i.e.,
	\[
	\hat{\phi}_{t, 1-\tau}^{+}  :=\hat{\phi}_{t}^{+}+\text{se}(\hat{\phi}_{t}^{+})\Phi^{-1}(1-\tau), \hat{\phi}_{t,\tau}^{-}  :=\hat{\phi}_{t}^{-}-\text{se}(\hat{\phi}_{t}^{-})\Phi^{-1}(1-\tau),
	\]where $\text{se}(\hat{\phi}_{t}^{+}):=\sqrt{\widehat{\text{Var}}(\hat{\phi}_{t}^{+})/n}$ and $\text{se}(\hat{\phi}_{t}^{-}):=\sqrt{\widehat{\text{Var}}(\hat{\phi}_{t}^{-})/n}$ denote the standard errors of $\hat{\phi}_{t}^{+}$ and $\hat{\phi}_{t}^{-}$. The following one-sided covering properties hold:
	\[
	\lim_{n\rightarrow\infty}P(\phi_{t}^{+}\leq\hat{\phi}_{t,1-\tau}^{+})\geq 1-\tau,  \lim_{n\rightarrow\infty}P(\phi_{t}^{-}\geq\hat{\phi}_{t,\tau}^{-})\geq 1-\tau.
	\]Suppose 
	\[
	\left\{t\in \boldsymbol{\Theta}_0: \hat{\phi}_{t,1-\tau}^{+}\geq0 \right\}\neq \emptyset, \left\{t\in \boldsymbol{\Theta}_0:\hat{\phi}_{t,\tau}^{-}\leq0\right\}\neq \emptyset.
	\]
	If $\theta = \theta_0$,
	\begin{equation}
		P\left(\theta_{0}\in\left\{t\in \boldsymbol{\Theta}_0: \hat{\phi}_{t,1-\tau}^{+}\geq0 \right\} \right) \geq 1-\tau, P\left(\theta_{0}\in\left\{t\in \boldsymbol{\Theta}_0:\hat{\phi}_{t,\tau}^{-}\leq0\right\} \right)\geq 1-\tau.
	\end{equation}
\end{theorem}

Define
\[
\text{CI}_{1-2\tau}^{[\lambda^{-},\lambda^{+}]}:=[\hat{\lambda}_{\tau}^{-},\hat{\lambda}_{1-\tau}^{+}],\text{CI}_{1-2\tau}^{[\gamma^{-},\gamma^{+}]}:=[\hat{\gamma}_{\tau}^{-},\hat{\gamma}_{1-\tau}^{+}],
\text{CI}_{1-2\tau}^{[\phi_{t}^{-},\phi_{t}^{+}]}:=\left[\hat{\phi}_{t,\tau}^{-},\hat{\phi}_{t,1-\tau}^{+}\right].
\]
We then have the following results for coverage probability of the partial identification regions and interested parameters.
\begin{theorem}
	\begin{align*}
		\lim_{n\rightarrow\infty}P(\lambda\in\text{CI}_{1-2\tau}^{[\lambda^{-},\lambda^{+}]})\geq\lim_{n\rightarrow\infty}P([\lambda^{-},\lambda^{+}]\in\text{CI}_{1-2\tau}^{[\lambda^{-},\lambda^{+}]})\geq 1-2\tau,\\
		\lim_{n\rightarrow\infty}P(\gamma\in\text{CI}_{1-2\tau}^{[\gamma^{-},\gamma^{+}]})\geq\lim_{n\rightarrow\infty}P([\gamma^{-},\gamma^{+}]\in\text{CI}_{1-2\tau}^{[\gamma^{-},\gamma^{+}]})\geq 1-2\tau,\\
		\lim_{n\rightarrow\infty}P(\phi_{t}\in\text{CI}_{1-2\tau}^{[\phi_{t}^{-},\phi_{t}^{+}]})\geq\lim_{n\rightarrow\infty}P\left(\left[{\phi}_{t}^{-},\phi_{t}^{+}\right]\in\text{CI}_{1-2\tau}^{[\phi_{t}^{-},\phi_{t}^{+}]}\right)\geq 1-2\tau.
	\end{align*}
\end{theorem}

As the product of the sensitivity parameters satisfies $\zeta_{Y,\alpha}\rightarrow 0$, the interval $\text{CI}_{1-2\tau}^{[\lambda^{-},\lambda^{+}]}$ converges to the (conventional) $(1-2\tau)$ C.I. for the short version parameter $\lambda_{s}$. Similarly, as $\zeta_{D,\alpha}\rightarrow 0$, $\text{CI}_{1-2\tau}^{[\gamma^{-},\gamma^{+}]}$ converges to the $(1-2\tau)$ C.I. for $\gamma_{s}$. For $\phi_{t}$ and $\theta_0$, assume that 
\[
\left\{t\in \boldsymbol{\Theta}_0: \hat{\phi}_{t,1-\tau}^{+}\geq0 \right\}\cap \left\{t\in \boldsymbol{\Theta}_0
:\hat{\phi}_{t,\tau}^{-}\leq0\right\}\neq \emptyset.
\]
Define 
\begin{equation}
	\hat{\theta}_{0,1-\tau}^{+}:=\sup_{t\in \boldsymbol{\Theta}_0}\left\{t: \hat{\phi}_{t,1-\tau}^{+}\geq0 \right\}\cap \left\{t:\hat{\phi}_{t,\tau}^{-}\leq0\right\},
	\hat{\theta}_{0,\tau}^{-}:=\inf_{t\in\boldsymbol{\Theta}_0}\left\{t: \hat{\phi}_{t,1-\tau}^{+}\geq0 \right\}\cap \left\{t:\hat{\phi}_{t,\tau}^{-}\leq0\right\},\label{theta0_ovb_CI}
\end{equation} When both $\zeta_{Y,\alpha}\rightarrow 0$ and $\zeta_{D,\alpha}\rightarrow 0$ hold, $\text{CI}_{1-2\tau}^{[\phi_{t}^{-},\phi_{t}^{+}]}$
converges to the $(1-2\tau)$ C.I. for $\lambda_{s}-\gamma_{s}t$, and $\left[\hat{\theta}_{0,\tau}^{-},\hat{\theta}_{0,1-\tau}^{+}\right]$ converges to the interval obtained by inverting the positive and negative parts of the $(1-2\tau)$ C.I. for $\lambda_{s}-\gamma_{s}t$.

\subsection{OVB-Adjusted Confidence Interval for the Interested Parameter}
Theorem 4 appears to suggest that $\text{CI}_{1-2\tau}^{[\lambda^{-},\lambda^{+}]}$,
$\text{CI}_{1-2\tau}^{[\gamma^{-},\gamma^{+}]}$ and $\text{CI}_{1-2\tau}^{[\phi_{t}^{-},\phi_{t}^{+}]}$,
which provide valid $(1-2\tau)$ coverage for the OVB
bounds, can be used directly as ($1-2\tau$) C.I.'s for $\lambda$, $\gamma$
and $\phi_{t}$ after accounting for OVB. However, these OVB-adjusted C.I.'s can be further improved to yield intervals. To illustrate this, we use $\lambda$ as an example (and the relevant results also hold for $\gamma$ and $\phi_{t}$). 

Let $\Delta_{\lambda}=\lambda^{+}-\lambda^{-}$
denote the identification region of $\lambda$. In our framework, this is a constant and is not affected by sample size. If $\Delta_{\lambda}$ has a strictly positive length (i.e., $\lambda^{+}>\lambda^{-}$), it can be shown that
\begin{align}
	P(\lambda\in\text{CI}_{1-2\tau}^{[\lambda^{-},\lambda^{+}]})= & P(\{ \lambda\geq\hat{\lambda}_{\tau}^{-}\} \cap\{ \lambda\leq\hat{\lambda}_{1-\tau}^{+}\} )\nonumber \\
	= & 1-P(\{ \lambda<\hat{\lambda}_{\tau}^{-}\} \cup\{ \lambda>\hat{\lambda}_{1-\tau}^{+}\} )\nonumber \\
	\geq & 1-P(\lambda<\hat{\lambda}_{\tau}^{-})-P( \lambda>\hat{\lambda}_{1-\tau}^{+})\nonumber \\
	= & 1-(1-P( \lambda\geq\hat{\lambda}^{-}-\Phi^{-1}(1-\tau)\text{se}(\hat{\lambda}^{-})))-(1-P(\lambda\leq\hat{\lambda}^{+}+\Phi^{-1}(1-\tau)\text{se}(\hat{\lambda}^{+})))\label{key_ovb_adjust_CI_condition}\\
	\rightarrow & \begin{cases}
		1 & \text{if }\lambda^{-}<\lambda<\lambda^{-},\\
		1-\tau & \text{if }\lambda=\lambda^{-},\\
		1-\tau & \text{if }\lambda=\lambda^{+}.
	\end{cases}\nonumber 
\end{align}
as $n\rightarrow\infty$. Hence \[\lim_{n\rightarrow\infty}P(\lambda\in\text{CI}_{1-2\tau}^{[\lambda^{-},\lambda^{+}]})\geq1-\tau.\]
This result, established in a more general form by \citet{IM_2004}, implies that a $(1-2\tau)$ C.I. for the OVB bounds $[\lambda^{-},\lambda^{+}]$
can be served as a $(1-\tau)$ C.I. for the true parameter $\lambda$.

The key condition for this result is the strict positivity of $\Delta_{\lambda}$. When $\Delta_{\lambda}$
is strictly positive, the true parameter $\lambda$ cannot be near both the upper bound $\lambda^{+}$ and lower bound $\lambda^{-}$ simultaneously. If $\lambda$ lies strictly inside the bounds, the
non-coverage risk converges to zero asymptotically. If $\lambda$ is
close to the lower (upper) bound, the risk that it exceeds the upper bound (falls short of the lower bound) is asymptotically negligible. Thus at least one of the second and third terms in (\ref{key_ovb_adjust_CI_condition}) vanishes asymptotically, ensuring coverage no smaller than the one-sided case. 

However, $\text{CI}_{1-2\tau}^{[\lambda^{-},\lambda^{+}]}$ is not a uniformly valid $(1-\tau)$ C.I. for $\lambda$ over $\Delta_{\lambda}$. For instance, in the extreme
scenario where $\Delta_{\lambda}=0$ (which cannot be ruled out in our case), $\lambda$ is point-identified ($\lambda=\lambda^{+}=\lambda^{-}$)
and $\text{CI}_{1-2\tau}^{[\lambda^{-},\lambda^{+}]}$ reverts to a conventional $(1-2\tau)$ C.I.: $\lim_{n\rightarrow\infty}P(\lambda\in\text{CI}_{1-2\tau}^{[\lambda^{-},\lambda^{+}]})\geq1-2\tau$. That is, the asymptotic non-coverage risk becomes $1-2\tau$, making the C.I. too narrow. 

To construct a uniformly valid $(1-\tau)$ OVB-adjusted C.I., we adopt the shrinkage method of \citet{Stoye_2009}, which ensures that the C.I. is not narrower than the conventional C.I. as $\Delta_{\lambda}$ is close to zero. \citet{IM_2004} also provided a method to avoid the above difficulty based on the estimated size of the identification region $\hat{\Delta}_{\lambda}$ (see equation (7) in \citet{IM_2004}). However, \citet{Stoye_2009} pointed out that Imben and Manski's approach implicitly relies on the superefficiency of $\hat{\Delta}_{\lambda}$: the estimation error of $\hat{\Delta}_{\lambda}$
must vanish fast enough (at least with rate $o_{p}(n^{-1/2})$) when
$\Delta_{\lambda}$ is near zero. Such a requirement is too restrictive for our proposed DML estimator.

The shrinkage method of \citet{Stoye_2009} adjusts the estimator of the identification region and allows the estimation error of $\hat{\Delta}_{\lambda}$ can be with order $O_{p}(n^{-1/2})$. We again use $\lambda$ as an example to
illustrate how to use this method to construct the OVB-adjusted C.I.
for $\lambda$ as follows. Suppose there exists a sequence $\vartheta_{n}$ such
that $\vartheta_{n}\rightarrow0$ and $\sqrt{n}\vartheta_{n}\rightarrow\infty$. Define
\[
\hat{\Delta}_{\lambda}^{*}:=\begin{cases}
	\hat{\Delta}_{\lambda}, & \text{if }\hat{\Delta}_{\lambda}>\vartheta_{n},\\
	0, & \text{Otherwise},
\end{cases}
\]
and 
\[
\hat{\rho}:=\widehat{\text{Corr}}(\hat{\lambda}^{+},\hat{\lambda}^{-})=\frac{\widehat{\text{Cov}}(\hat{\lambda}^{+},\hat{\lambda}^{-})}{\sqrt{\widehat{\text{Var}}(\hat{\lambda}^{+})}\sqrt{\widehat{\text{Var}}(\hat{\lambda}^{-})}}.
\]The sequence $\vartheta_{n}$ controls the degree of shrinkage imposed on the initial estimator of the identification region. A slower decay of $\vartheta_{n}$ to zero indicates a less distortion from the shrinkage but a worsen accuracy of using the uniform normal approximation. The critical values for constructing the $(1-\tau)$ OVB-adjusted C.I., denoted by $z_{l}^{*}$ and $z_{u}^{*}$, are obtained from solving the following constrained minimization problem: 
\[
\min_{z_{l},z_{u}}z_{l}\sqrt{\widehat{\text{Var}}(\hat{\lambda}^{-})}+z_{u}\sqrt{\widehat{\text{Var}}(\hat{\lambda}^{+})}
\]
subject to 
\begin{align*}
	P\left(-z_{l}\leq Z_{1},\hat{\rho}Z_{1}\leq z_{u}+\frac{\hat{\Delta}_{\lambda}^{*}}{\text{se}(\hat{\lambda}^{+})}+\sqrt{1-\hat{\rho}^{2}}Z_{2}\right) & \geq1-\tau,\\
	P\left(-z_{l}-\frac{\hat{\Delta}_{\lambda}^{*}}{\text{se}(\hat{\lambda}^{-})}+\sqrt{1-\hat{\rho}^{2}}Z_{2}\leq\hat{\rho}Z_{1},Z_{1}\leq z_{u}\right) & \geq1-\tau,
\end{align*}
where $(Z_{1},Z_{2})$ are two independent standard normal random
variables. Then applying Proposition 3 of \citet{Stoye_2009} yields:
\[
\lim_{n\rightarrow\infty}P(\lambda\in\text{CI}_{1-\tau}^{\lambda,*})\geq1-\tau,
\]
where 
\[
\text{CI}_{1-\tau}^{\lambda,*}=\begin{cases}
	\left[\hat{\lambda}^{+}-\text{se}(\hat{\lambda}^{-})z_{l}^{*},\hat{\lambda}^{+}+\text{se}(\hat{\lambda}^{+})z_{u}^{*}\right], & \text{if }\hat{\lambda}^{-}-\text{se}(\hat{\lambda}^{-})z_{l}^{*}\leq\hat{\lambda}^{+}+\text{se}(\hat{\lambda}^{+})z_{u}^{*},\\
	\emptyset, & \text{Otherwise},
\end{cases}
\]
is the $(1-\tau)$ OVB-adjusted C.I. for $\lambda$. Similar procedures
can be applied to construct the $(1-\tau)$ OVB-adjusted C.I. for $\gamma$ and $\phi_{\theta_{0}}$. For $\theta_{0}$, we first can have that for $t\in\boldsymbol{\Theta}_{0}$,
\[\lim_{n\rightarrow\infty}P(\phi_{t}\in\text{CI}_{1-\tau}^{\phi_{t},*})\geq1-\tau,\]
where 
\[
\text{CI}_{1-\tau}^{\phi_{t},*}=\begin{cases}
	\left[\hat{\phi}_{t}^{-}-\text{se}(\hat{\phi}_{t}^{-})z_{l}^{\prime*},\hat{\phi}_{t}^{+}+\text{se}(\hat{\phi}_{t}^{+})z_{u}^{\prime*}\right], & \text{if }\hat{\phi}_{t}^{-}-\text{se}(\hat{\phi}_{t}^{-})z_{l}^{\prime*}\leq\hat{\phi}_{t}^{+}+\text{se}(\hat{\phi}_{t}^{+})z_{u}^{\prime*},\\
	\emptyset, & \text{Otherwise},
\end{cases}
\]
and $(z_{l}^{\prime*},z_{u}^{\prime*})$ are the solutions
for the corresponding constrained minimization problem of Stoye's
shrinkage method for $\phi_{t}$. If $\theta=\theta_{0}$, $\phi_{\theta_{0}}=0$,
then we have:
\begin{equation}
	\lim_{n\rightarrow\infty}P\left(\theta_{0}\in\left\{ t\in\boldsymbol{\Theta}_{0}:0\in\text{CI}_{1-\tau}^{\phi_{t},*}\right\} \right)\geq1-\tau.\label{OVB_adj_bound_theta0}
\end{equation}
Therefore $(1-\tau)$ OVB-adjusted C.I. for $\theta_{0}$ can be constructed accordingly. 

\subsection{K-Fold Cross-Fitting and the Median Method}
In this section, we introduce two methods which can be adopted to further mitigate overfitting bias caused by machine learning estimators: K-fold cross-fitting and the median method \citep{CCDDHNR_2018}. K-fold cross-fitting uses different parts of samples to repeatedly estimate and predict parameters: $(\lambda_s, \gamma_s,v^{2}_s,\sigma^{2}_{Y_s},\sigma^{2}_{D_s})$, and take an average of the predictions to form the final estimate of the parameter. Let $W_{s,i}=(Z_i,X_i)$ and $V_{s,i}=(Y_{i},D_{i},W_{s,i})$ denote the $i$-th available observation, $i=1,2,\ldots,n$. Below we illustrate the procedures for conducting the K-fold cross fitting on estimating $\theta_s=\lambda_s/\gamma_s$ with estimators based on (\ref{LATE_DML}): 
\begin{enumerate}
	\item Randomly split the $n$ samples into $K$ (mutually exclusive) subsamples
	of equal sample size $n_{k}=n/K$, $k=1,2,\ldots,K$.
	Let $I_{k}$, $k=1,2,\ldots,K$ denote the set of indices for the
	$K$ different subsamples. Let $I_{k}^{c}$, $k=1,2,\ldots,K$ denote
	the complement set of $I_{k}$: $I_{k}^{c}=\left\{ 1,2,\ldots,n\right\} \setminus I_{k}$.
	\item For each $k$, estimate models of the nuisance parameters $\pi_s(X)$, $E[Y|W_s]$ and $E[D|W_s]$ with the available observations $V_{s,i}$,
	$i\in I_{k}^{c}$. Using the available observations $V_{s,i}$, $i\in I_{k}$ to obtain predictions of the nuisance parameters: 
	$\hat{\pi}_s^{(k)}(X_i)$, $\hat{E}^{(k)}[Y_i|Z_i=1,X_i]$, $\hat{E}^{(k)}[Y_i|Z_i=0,X_i]$, $\hat{E}^{(k)}[D_i|Z_i=1,X_i]$ and $\hat{E}^{(k)}[D_i|Z_i=0,X_i]$. 
	\item For each $k$, compute the estimate of $\lambda_s$ and $\gamma_s$
	using the predicted nuisance parameters of step 2 as
	\begin{equation*}
		\lambda_{s}^{(k)} = \frac{1}{n_{k}}\sum_{i\in I_{k}}\hat{\bar{\psi}}_{\lambda_{s}}^{(k)}(Y_{i},W_{s,i}),\text{ }
		\gamma_{s}^{(k)} = \frac{1}{n_{k}}\sum_{i\in I_{k}}\hat{\bar{\psi}}_{\gamma_{s}}^{(k)}(D_i,W_{s,i}),
	\end{equation*}where
	\begin{eqnarray*}
		\hat{\bar{\psi}}_{\lambda_{s}}^{(k)}(Y_{i},W_{s,i}) & = & \frac{Z_{i}}{\hat{\pi}_{s}^{(k)}(X_{i})}(Y_{i}-\hat{E}^{(k)}[Y_{i}|Z_{i}=1,X_{i}])-\frac{1-Z_{i}}{1-\hat{\pi}_{s}^{(k)}(X_{i})}(Y_{i}-\hat{E}^{(k)}[Y_{i}|Z_{i}=0,X_{i}])+\\
		&  & \hat{E}^{(k)}[Y_{i}|Z_{i}=1,X_{i}]-\hat{E}^{(k)}[Y_{i}|Z_{i}=0,X_{i}],\\
		\hat{\bar{\psi}}_{\gamma_{s}}^{(k)}(D_i,W_{s,i}) & = & \frac{Z_{i}}{\hat{\pi}_{s}^{(k)}(X_{i})}(D_{i}-\hat{E}^{(k)}[D_{i}|Z_{i}=1,X_{i}])-\frac{1-Z_{i}}{1-\hat{\pi}_{s}^{(k)}(X_{i})}(D_{i}-\hat{E}^{(k)}[D_{i}|Z_{i}=0,X_{i}])+\\
		&  & \hat{E}^{(k)}[D_{i}|Z_{i}=1,X_{i}]-\hat{E}^{(k)}[D_{i}|Z_{i}=0,X_{i}].
	\end{eqnarray*}\item Average $\hat{\lambda}_{s}^{(k)}$ and $\hat{\gamma}_{s}^{(k)}$ over $k=1,2,\ldots,K$
	to obtain the estimates of $\lambda_s$ and $\gamma_s$:
	\[\hat{\lambda}_s=\frac{1}{K}\sum_{k=1}^{K}\hat{\lambda}_{s}^{(k)}, \hat{\gamma}_s=\frac{1}{K}\sum_{k=1}^{K}\hat{\gamma}_{s}^{(k)}.\]The estimate of $\theta_s$ is:
	\begin{equation}
		\hat{\theta}_s = \frac{\hat{\lambda}_s}{\hat{\gamma}_s}.\label{K_fold_theta_s}
	\end{equation}
\end{enumerate}

$\hat{\theta}_{s}$ in (\ref{K_fold_theta_s}) is referred to as DML2 estimator in \citet{CCDDHNR_2018} and is the equivalent to solving $\theta_{s}$ in the following equation: \[\frac{1}{K}\sum_{k=1}^{K}\left[\frac{1}{n_{k}}\sum_{i\in I_{k}}\hat{\bar{\psi}}_{\theta_{s}}^{(k)}(Y_{i},D_{i},W_{s,i})\right]=0,\]where \[\hat{\bar{\psi}}_{\theta_{s}}^{(k)}(Y_{i},D_{i},W_{s,i})=\hat{\bar{\psi}}_{\lambda_{s}}^{(k)}(Y_{i},W_{s,i})-\hat{\bar{\psi}}_{\gamma_{s}}^{(k)}(D_{i},W_{s,i})\theta_{s}.\]Alternatively, we may estimate $\theta_{s}$ by using DML1 estimator: $\hat{\theta}_{s}^{\prime}=1/K\sum_{k=1}^{K}\hat{\lambda}_{s}^{(k)}/\hat{\gamma}_{s}^{(k)}$, that is, taking an average of $\hat{\lambda}_{s}^{(k)}/\hat{\gamma}_{s}^{(k)}$, $k=1,\ldots,K$. In practice, DML2 estimator is more preferred than DML1 estimator, since the former generally has a more stable property than the latter and therefore demonstrates a better performance empirically \citep{CCDDHNR_2018}.

Furthermore, to avoid uncertainty from sample splitting in the K-fold cross-fitting, we adopt the median method suggested in \citet{CCDDHNR_2018} to improve stability of our final estimates of $(\lambda_{s},\gamma_{s},v_{s}^{2},\sigma_{Y_{s}}^{2},\sigma_{D_{s}}^{2})$. To implement the median method, we first repeat the procedures of the K-fold cross-fitting $L$ times. 
Let $\hat{\boldsymbol{\xi}}^{l}$ denote a vector of the estimated parameters and $\hat{\mathbf{\Sigma}}^{l}$ denote the estimated approximate covariance matrix of $\sqrt{n}(\hat{\boldsymbol{\xi}}^{l}-\boldsymbol{\xi})$ (i.e., $\boldsymbol{\Omega}$ in Theorem 2), from the $l$th K-fold cross-fitting, $l=1,\ldots,L$. We use 
\begin{equation}
	\hat{\boldsymbol{\xi}}^{\text{Median}}=\text{Median}\{\hat{\boldsymbol{\xi}}^{l}\}_{l=1}^{L}\label{med_par}
\end{equation}as the final estimate of the parameters and
\begin{equation}
	\hat{\mathbf{\Sigma}}^{\text{Median}} = \text{Median}\{\hat{\mathbf{\Sigma}}^{l}+(\hat{\boldsymbol{\xi}}^{l}-\hat{\boldsymbol{\xi}}^{\text{Median}})(\hat{\boldsymbol{\xi}}^{l}-\hat{\boldsymbol{\xi}}^{\text{Median}})^{\top}\}_{l=1}^{L}\label{med_app_cov}
\end{equation}as the final estimate of the approximate covariance matrix.\footnote{The ``Median'' in (\ref{med_par}) chooses the median among the $L$ cross fittings for each of the estimated parameters , while in (\ref{med_app_cov}), it chooses the matrix with median operator norm.} 

\section{Empirical Application with the JTPA Data}
In this section, we demonstrate the usefulness of our proposed method for quantifying the OVB in nonlinear IV estimators through an empirical application. We perform an OVB analysis for LATE and LATT estimations in Title II programs of the Job Training Partnership Act (JTPA) in the US. The data consist of adult male and female workers who participated in these programs between November 1987 and September 1989. 
Following \citet{AAI_2002}, we assume the observations are i.i.d. for estimation purposes. The outcome variable $Y$ is the total earnings in the 30 months. 
The treatment variable $D$ is a binary variable for enrollment in the JTPA services (1 = enrolled; 0 = not enrolled), while the instrumental variable $Z$ indicates whether the individual was offered such services (1 = offered; 0 = not offered). The exogenous covariates include age (\texttt{age}), which is a categorical variable, as well as a set of dummy variables: black (\texttt{black}), Hispanic (\texttt{hispanic}), high-school graduates (\texttt{hsorged}, including GED holders), marital status (\texttt{married}), AFDC receipt (\texttt{adfc}, for adult female workers only), whether the applicant worked at least 12 weeks in the 12 months preceding random assignment (\texttt{wkless13}), the original recommended service strategy: classroom training (\texttt{class\_tr}), and OJT/JSA/other (\texttt{ojt\_jsa}), and whether earnings data were from the second follow-up survey (\texttt{f2sms}). 
The total sample size is 11,204 (5,102 males and 6,102 females).


Although offers for the JTPA services were randomly assigned, only approximately 60\% of those offered the services actually enrolled \citep{AAI_2002}. This partial compliance raises potential endogeneity concerns, as treatment status may be self-selected and correlated with the potential outcomes. Since the offers were randomly assigned and were considered to likely influence participants' intention to enroll in the program, we use the offer assignment as the instrumental variable. While some individuals received services without being assigned, \cite{AAI_2002} note that this violation was rare (less than 2\%) and thus unlikely to materially affect our estimates.

For LATE, Figure \ref{figure1} and \ref{figure2} present sensitivity contour plots for the lower bounds of the 97.5\% confidence intervals (C.I.s) for $\lambda^{-}$ (left panel) and $\gamma^{-}$ (right panel) assuming $|\rho_Y|=|\rho_D|=1$. Figure \ref{figure1} corresponds to male workers and Figure \ref{figure2} to female workers. Each contour line shows the lower bound of the 97.5\% C.I. when the product of $C_YC_{\alpha}$ (or $C_DC_{\alpha}$) equals to a specific value. For instance, consider the case of male workers. When $C_YC_{\alpha}$=4.55e-3 (i.e., $C_Y=0.1$ and $C_{\alpha} = 0.0455$), the contour line indicates that the lower bound of the 97.5\% confidence interval for $\lambda^{-}$ roughly equals to -300.

The sensitivity parameters $C_Y C_{\alpha}$ and $C_D C_{\alpha}$ have a negative impact on the lower bounds of the C.I.s for $\lambda^{-}$ and $\gamma^{-}$. 
For the male workers, when $C_YC_{\alpha}=0$ (either $C_Y=0$ or $C_{\alpha} = 0$, or both), the lower bound is -114.99, suggesting that even without considering the OVB, the (short) estimate of $\lambda$ (ITT) is not statistically significant at the 5\% level. In fact, the value for $C_Y C_{\alpha}$ to push the lower bound below zero is a slightly negative (roughly equals to -0.003), which is not feasible, since $C_Y C_{\alpha}$ is required to be nonnegative. For female workers, the corresponding threshold for $C_Y C_{\alpha}$ roughly equal to 0.019. In contrast, for $\gamma^{-}$, as shown in the right panels of Figures \ref{figure1} and \ref{figure2}, the thresholds are much less stringent than those for $\lambda^{-}$. For both male and female workers, the criteria both exceed 0.65, indicating that estimates of $\gamma$ ($P(T=C)$) are much robust to the omitted variables compared to the estimates of $\lambda$ (ITT).   

We next turn to the results of LATE estimation when considering the OVB. These results are obtained using the calibrated sensitivity parameters $C_{\alpha}$, $C_Y$ and $C_D$. To determine the calibrated value, we first estimate the sensitivity parameter separately for each covariate using the method introduced in the benchmarking analysis. We set the relative strength indicator $k_{\alpha} =k_Y= k_D = 1$, which implies that the omitted variable is assumed to be at least as important as any excluded covariate $X_j$ in predicting $(Y,D,Z)$, given the remaining covariates $X_{-j}$. In addition, we set $|\rho_Y| = |\rho_D| = 1$, which yields the maximal values of the estimated sensitivity parameters. 
Then the largest among these estimates is selected as the calibrated value of the sensitivity parameter. This maximum estimate (calibrated value of the sensitivity parameter) and associated covariate (denoted by $X_j^{*}$) are reported in Table \ref{table1}.\footnote{Here, the variable \texttt{age} represents all age-related categorical variables.} For LATE, the results indicate that if the omitted variable is as important as $X_j^{*}$, including it would enhance prediction precision for $Z$ by 1.9\% ($0.138^2$) among male workers and 0.62\% ($0.079^2$) among female workers. In the case of male workers, the reduction in MSE when predicting $Y$ and $D$ would be 2.2\% and 0.18\%, while for female workers, the corresponding reductions are 3.2\% and 0.35\%. 
Overall, the estimated influence of the omitted variable on the prediction of $(Y,D,Z)$ appears to be small in the case of LATE estimation. 

The first three columns of Table \ref{table2} present short estimates (those estimated with the available data) and their corresponding 95\% C.I.s, and estimates of the OVB bounds for the parameters. 
Estimates of the OVB bounds for $\lambda$ and $\gamma$ are estimates of $(\lambda^{+},\lambda^{-})$ and $(\gamma^{+},\gamma^{-})$, while estimates of the OVB bounds for $\theta$ are estimates obtained from using the derived result in Theorem 1. 
From the table, the short estimates of $\gamma$ ($P(T=C)$) are statistically significant at the 5\% level for both male and female workers. However, the significance of $\lambda$ (ITT) and $\theta$ (LATE) differs by gender: for female workers, both estimates remain statistically significant at the 5\% level, while for male workers, they do not. 
When accounting for the OVB based on the calibrated sensitivity parameters, estimates of the OVB bounds suggest that the true LATE for the male workers lies within the range of 317 to 3,044 U.S. dollars, and for the female workers, the range is within 1,279 to 2,531 U.S. dollars. 

The last column shows the estimated bounds $[\hat{\theta}^{-}_{0,\tau},\hat{\theta}^{+}_{0,1-\tau}]$ in (\ref{theta0_ovb_CI}),  $[\hat{\lambda}^{-}_{\tau},\hat{\lambda}^{+}_{1-\tau}]$ in (\ref{lambda_ovb_CI}) and $[\hat{\gamma}^{-}_{\tau},\hat{\gamma}^{+}_{1-\tau}]$ in (\ref{gamma_ovb_CI}) with $\tau = 0.025$, which are denoted by $\text{Low}_{0.025}$ and $\text{Up}_{0.975}$ in the table. Through Figure \ref{figure3}, we illustrate how to practically use (\ref{theta0_ovb_CI}) to obtain $[\hat{\theta}^{-}_{0,\tau},\hat{\theta}^{+}_{0,1-\tau}]$. Figure \ref{figure3} plots the estimated functions $\hat{\phi}_{t,0.975}^{+}$ and $\hat{\phi}_{t,0.025}^{-}$ over different values of $t$ based on the calibrated sensitivity parameters. Each function is segmented into two parts: one for $t\geq 0$ (solid line) and one for $t<0$ (dashed line). 
The two estimated functions are generally continuous in $t$ but not
differentiable at $t=0$. Within the selected range of $t$ in the plots, both segments of the estimated functions decrease monotonically. 

For the male workers, the plot shows that $\hat{\phi}_{t,0.975}^{+}\geq0$ if $t\leq4,916.2$, and
$\hat{\phi}_{t,0.025}^{-}\leq0$ if $t\geq-1,551.26$. According to
(\ref{theta0_ovb_CI}), this implies that $[\hat{\theta}^{-}_{0,0.025},\hat{\theta}^{+}_{0,0.975}]= [-1,551.26, 4,916.20]$. For the female workers, applying the same logic yields $[\hat{\theta}^{-}_{0,0.025},\hat{\theta}^{+}_{0,0.975}]= [198.50, 3,625.84]$ based on the corresponding calibrated sensitivity parameters. From the results, it can be seen that as the uncertainty associated with OVB bound estimation is incorporated, the statistical significance results are align with those based on the point estimates. In particular, for female workers, the statistical significance of LATE (and ITT) estimate remains robustly stand after accounting for the OVB and uncertainty of the estimation.


For LATT, the relevant results are shown in Table \ref{table1} and \ref{table3} and Figure \ref{figure4} to \ref{figure6}. The results are qualitatively similar as those for LATE. Specifically, the estimates of $\gamma$ are much robust to the omitted variables than the estimates of $\lambda$. For male workers, the threshold for $C_Y C_{\alpha}$ that brings the lower bound below zero roughly equals to -0.004. For female workers, the corresponding threshold is approximately 0.02. In contrast, the thresholds for $\gamma^{-}$, shown in the right panels of Figures \ref{figure4} and \ref{figure5}, are much less stringent than those for $\lambda^{-}$. For both male and female workers, these thresholds are above 0.64, reinforcing the robustness of the estimated $\gamma$ to the omitted variables. Results of the OVB-adjusted estimates are again align with those of the point estimates. Importantly, the statistical significance of both $\lambda$ and LATT estimates for the female workers remains robust even after accounting for both the OVB and uncertainty in the estimation of its bounds.

\subsection{Statistical Significance after Accounting for the OVB} 
Table \ref{table4} to \ref{table5} present the results for the $(1-\tau)$ OVB-adjusted confidence intervals for LATE and LATT constructed using the shrinkage method
of \citet{Stoye_2009}. We set the significance level $\tau=0.05$ (i.e., 95\% C.I.)
and the shrinkage factor as
\begin{equation}
\vartheta_{n}=\sqrt{\frac{\log\log n}{n}}\times\max\{\hat{\sigma}_{l},\hat{\sigma}_{u}\},\label{shrinkage_factor_Stoye09}
\end{equation}
where $\hat{\sigma}_{l}$ and $\hat{\sigma}_{u}$ denote the estimated standard deviations of the lower and upper OVB bounds. The shrinkage factor (\ref{shrinkage_factor_Stoye09}) is the one (from the iterated law of logarithm) suggested in \citet{Stoye_2009}, scaled by $\max\{\hat{\sigma}_{l},\hat{\sigma}_{u}\}$.    

For $\lambda$ and $\gamma$, $(z_{l}^{*},z_{u}^{*})$ denote the critical values, $\hat{\Delta}^{*}$ denote the estimate of the identification region, and Min.Obj. denote the minimum value of the constrained minimization for Stoye's shrinkage method. To construct the $(1-\tau)$ OVB-adjusted C.I. for $\theta_{0}$ (LATE or LATT), we first compute $\text{C.I.}_{1-\tau}^{\phi_{t},*}$ for $t$ over a specified range, and then obtain the upper and lower bounds of the C.I. by inverting (\ref{OVB_adj_bound_theta0}). For $\theta_{0}$, the reported values of $(z_{l}^{*},z_{u}^{*})$, $\hat{\Delta}^{*}$ and Min.Obj. correspond to those for $\phi_{t}$, averaged over different values of $t$. Figure \ref{figure7} and \ref{figure8} show plots of the upper and lower bounds of $\text{C.I.}_{1-\tau}^{\phi_{t},*}$
as functions of $t$ for LATE and LATT. 

From the two tables, we observe that at the 95\% level, the statistical significance of $\theta_{0}$, $\lambda$ and $\gamma$
after accounting for OVB is qualitatively similar to the results obtained without OVB adjustment (i.e., the short estimates). For $\theta_{0}$ and $\lambda$, the 95\% OVB-adjusted
C.I.'s are narrower than the intervals $[\text{Low}_{0.025},\text{Up}_{0.975}]$
shown in previous tables. This is due to relatively large estimates of the identification regions for $\phi_{t}$ and $\lambda$, which lead to lower (one-sided) critical values. In our settings, the critical values used are almost equal to 1.645 or 1.96, corresponding to the 95\% or 97.5\% quantile of the standard normal distribution. This occurs because for each of $\phi_{t}$, $\lambda$ and $\gamma$, the estimated correlation between the estimated upper and lower OVB bounds is very close to one. As shown in \citet{Stoye_2009}, in such a situation, the solved critical values for the $(1-\tau)$ OVB-adjusted C.I. are very close to $(1-\tau)$ (or $(1-2\tau)$) quantile of the standard normal distribution.

\section{Conclusion}
This paper develops a general framework for quantifying omitted variable bias (OVB) in nonlinear instrumental variable (IV) estimators. Extending the recent work of \citet{CCNSS_2024}, we analyze a class of estimators — including the local average treatment effect (LATE), LATE on the treated (LATT), and the partially linear IV model (PLIVM) — that can be expressed as ratios of reduced-form and first-stage parameters. We derive bias decompositions for these parameters, establish partial identification bounds for the structural estimand, and construct statistical inference procedures that yield OVB-adjusted confidence intervals. Estimation is conducted via double machine learning and the median method \citep{CCDDHNR_2018}. An empirical application to the U.S. Job Training Partnership Act (JTPA) experiment shows that estimates of the first-stage probability of compliance are robust to omitted variables, while intention-to-treat and treatment effect estimates are more sensitive. Specifically, female participants exhibit robust and statistically significant program impacts, whereas effects for males become fragile once OVB is or nor accounted for. Overall, this study provides a unified framework for sensitivity analysis of nonlinear IV estimators and offers practical tools for assessing the robustness of causal conclusions in applied research.
\clearpage
\appendix
\section{Proof of Theorems}
\subsection{Proof of Theorem 1}
\begin{proof}
	If $\theta = \theta_0$, $\phi_{\theta_{0}}=0$. It follows that $0\in[\phi_{\theta_{0}}^{-},\phi_{\theta_{0}}^{+}]$, which implies that $\phi_{\theta_{0}}^{+}\geq0$ and $\phi_{\theta_{0}}^{-}\leq0$ both hold. This result can be used to obtain the OVB bound for $\theta_{0}$.
	Note that $\phi_{\theta_{0}}=0$ is equivalent to $\theta_{0}=\lambda/\gamma$.
	The OVB bound for $\theta_{0}$ thus depends on the partially identified sets for $(\lambda,\gamma)$ when OVB is present. On the other hand, by showing the possible range of $\theta_{0}$ when considering OVB of $(\lambda,\gamma)$, the OVB bound for $\theta_{0}$ can be established accordingly. 
	
	We proceed the proof by considering different sign
	scenarios for $(\gamma^{-},\gamma^{+})$ and $(\lambda^{-},\lambda^{+})$. We first show how to obtain the OVB bound for $\theta_{0}$ when $(\gamma^{-},\gamma^{+})\in\mathbb{R}^{++}$. In this scenario, $\gamma>0$.
	\begin{itemize}
		\item If $(\lambda^{-},\lambda^{+})\in\mathbb{R}^{++}$, $\lambda>0$ and hence $\theta_{0}>0$. Given $\phi_{\theta_{0}}^{+}\geq0$
		and $\phi_{\theta_{0}}^{-}\leq0$, we have $\lambda^{+}-\gamma^{-}\theta_{0}\geq0$
		and $\lambda^{-}-\gamma^{+}\theta_{0}\leq0$. Therefore $\theta_{0}\in[\lambda^{-}/\gamma^{+},\lambda^{+}/\gamma^{-}]$.
		\item If $(\lambda^{-},\lambda^{+})\in\mathbb{R}^{--}$, $\lambda<0$ which
		implies $\theta_{0}<0$. Again using the bound on $\phi_{\theta_{0}}$, we have $\lambda^{+}-\gamma^{+}\theta_{0}\geq0$
		and $\lambda^{-}-\gamma^{-}\theta_{0}\leq0$. Thus $\theta_{0}\in[\lambda^{-}/\gamma^{-},\lambda^{+}/\gamma^{+}]$.
		\item If $(\lambda^{-},\lambda^{+})$ have different signs, then $\lambda$ is not sign-deterministic. We derive the partially identified sets for $\theta_{0}$
		separately under $\lambda\geq0$ and $\lambda<0$, and then take their union
		as the OVB bound for $\theta_{0}$. From previous results, for $\lambda\geq0$,
		$\theta_{0}\in[\lambda^{-}/\gamma^{+},\lambda^{+}/\gamma^{-}]$; for
		$\lambda<0$, $\theta_{0}\in[\lambda^{-}/\gamma^{-},\lambda^{+}/\gamma^{+}]$.
		Therefore the OVB bound for $\theta_{0}$ in this
		situation is $[\lambda^{-}/\gamma^{+},\lambda^{+}/\gamma^{-}]\cup[\lambda^{-}/\gamma^{-},\lambda^{+}/\gamma^{+}]=[\lambda^{-}/\gamma^{-},\lambda^{+}/\gamma^{-}]$,
		which includes zero since $(\lambda^{-}/\gamma^{-},\lambda^{+}/\gamma^{-})$
		have different signs.
	\end{itemize}
	
	For $(\gamma^{-},\gamma^{+})\in\mathbb{R}^{--}$, 
	since arguments for the proof are very similar as those used in the
	scenario $(\gamma^{-},\gamma^{+})\in\mathbb{R}^{++}$, we omit it for brevity. 
	
	For $(\gamma^{-}, \gamma^{+})=(0,0)$, it can be shown that both the upper and lower OVB bounds for $\theta_{0}$ are undefined. Therefore the scenario is excluded. 
	
	We now turn to the scenario that $(\gamma^{-},\gamma^{+})$ have different signs. This case is more
	complex than those when $(\gamma^{-},\gamma^{+})$ have same sign,
	since (a) $\gamma$ is no longer sign-deterministic, and (b) the interval $[\gamma^{-},\gamma^{+}]$
	includes zero, so we need to separately consider the cases when $\gamma\rightarrow0^{-}$ and $\gamma\rightarrow0^{+}$. We start from the case when both $\gamma^{+}\neq 0$ and $\gamma^{-}\neq 0$. Note that in the following proof, we exclude the case when $\gamma=0$, since $\theta_{0}$ is not defined. 
	\begin{itemize}
		\item If $(\lambda^{-},\lambda^{+})\in\mathbb{R}^{++}$, $\lambda>0$.
		\begin{itemize}
			\item As $\gamma\rightarrow0^{-}$, $\theta_{0}\rightarrow-\infty$; as $\gamma\rightarrow0^{+}$, $\theta_{0}\rightarrow\infty$.
			\item When $\gamma>>0$, $\theta_{0}>0$.
			With $\phi_{\theta_{0}}^{+}\geq0$ and $\phi_{\theta_{0}}^{-}\leq0$,
			we have $\lambda^{+}-\gamma^{-}\theta_{0}\geq0$ and $\lambda^{-}-\gamma^{+}\theta_{0}\leq0$.
			Thus $\theta_{0}\geq\lambda^{-}/\gamma^{+}>0$ (since
			$\lambda^{-}/\gamma^{+}>0>\lambda^{+}/\gamma^{-}$). 
			\item When $\gamma<<0$,
			$\theta_{0}<0$. By using similar arguments, we have $\lambda^{+}-\gamma^{+}\theta_{0}\geq0$
			and $\lambda^{-}-\gamma^{-}\theta_{0}\leq0$. Thus $\theta_{0}\leq\lambda^{-}/\gamma^{-}<0$ (since $\lambda^{+}/\gamma^{+}>0>\lambda^{-}/\gamma^{-}$).
		\end{itemize}
		\item[] Summing the above results, we conclude that $\theta_{0}\in(-\infty,\lambda^{-}/\gamma^{-}]\cup[\lambda^{-}/\gamma^{+},\infty)$.
		\item If $(\lambda^{-},\lambda^{+})\in\mathbb{R}^{--}$, $\lambda<0$.
		\begin{itemize}
			\item As $\gamma\rightarrow0^{-}$ and $\gamma\rightarrow0^{+}$, $\theta_{0}\rightarrow-\infty$
			and $\theta_{0}\rightarrow\infty$. 
			\item When $\gamma>>0$, $\theta_{0}<0$. With the bound on $\phi_{\theta_{0}}$, we have $\lambda^{+}-\gamma^{+}\theta_{0}\geq0$ and $\lambda^{-}-\gamma^{-}\theta_{0}\leq0$. We conclude
			that $\theta_{0}\leq\lambda^{+}/\gamma^{+}<0$ (since
			$\lambda^{+}/\gamma^{+}<0<\lambda^{-}/\gamma^{-}$).
			\item When $\gamma<<0$, $\theta_{0}>0$. By using similar arguments, we have $\lambda^{+}-\gamma^{-}\theta_{0}\geq0$
			and $\lambda^{-}-\gamma^{+}\theta_{0}\leq0$. Therefore $\theta_{0}\geq\lambda^{+}/\gamma^{-}>0$ (since $\lambda^{-}/\gamma^{+}<0<\lambda^{+}/\gamma^{-}$).
		\end{itemize}
		\item[] Summing the above results, we
		conclude that $\theta_{0}\in(-\infty,\lambda^{+}/\gamma^{+}]\cup[\lambda^{+}/\gamma^{-},\infty)$.
		\item Now consider when $(\lambda^{-},\lambda^{+})$ have different signs.
		\begin{itemize}
			\item As $\gamma\rightarrow0^{-}$ and $\gamma\rightarrow0^{+}$, $\theta_{0}\rightarrow-\infty$
			and $\theta_{0}\rightarrow\infty$.
			\item When $\gamma>>0$ and $\lambda\geq 0$, we have $\theta_{0}\geq 0$.
			\item When $\gamma<<0$ and $\lambda\leq 0$, we have $\theta_{0}\geq 0$. 
			\item When $\gamma>>0$ and $\lambda<0$ or $\gamma<<0$ and $\lambda>0$, we have $\theta_{0}< 0$. 
		\end{itemize}
		\item[] Therefore in this scenario, we conclude that 
		\[\theta_{0}\in(-\infty,0]\cup[0,\infty)=(-\infty,\infty).\]
	\end{itemize}
	
	As for one of $\gamma^{+}$ and $\gamma^{-}$ equals to zero, it can be shown that one of the upper and lower OVB bounds for $\theta_{0}$ is undefined. Therefore these scenarios are excluded.
\end{proof}

\subsection{Proof of Theorem 2}
\begin{proof}
	Let $\widehat{\Xi}=(\hat{\lambda}_{s},\hat{\gamma}_{s},\hat{v}_{s}^{2},\hat{\sigma}_{Y_{s}}^{2},\hat{\sigma}_{D_{s}}^{2})$
	be the vector of DML estimators for the short version parameters $\Xi=(\lambda_{s},\gamma_{s},v_{s}^{2},\sigma_{Y_{s}}^{2},\sigma_{D_{s}}^{2})$.
	Let 
	\[
	\boldsymbol{\psi}(W_{s};\overline{\Xi})=\left[\psi_{\bar{\lambda}s},\psi_{\bar{\gamma}_{s}},\psi_{\bar{v}_{s}^{2}},\psi_{\bar{\sigma}_{Y_{s}}^{2}},\psi_{\bar{\sigma}_{D_{s}}^{2}}\right]^{\top},
	\]
	where $\overline{\Xi}=(\bar{\lambda}_{s},\bar{\gamma}_{s},\bar{v}_{s}^{2},\bar{\sigma}_{Y_{s}}^{2},\bar{\sigma}_{D_{s}}^{2})$.
	If Assumption 4.2 (for PLIVM) or Assumption 5.2 (for LATE) in \citet{CCDDHNR_2018} holds, it can be shown that 
	\[
	\sqrt{n}(\widehat{\Xi}-\Xi)\overset{d}{\rightarrow}N(0,\boldsymbol{\Omega}),
	\]
	where $\boldsymbol{\Omega}=\mathbf{J}_{0}^{-1}E[\boldsymbol{\psi}(W_{s};\Xi)\boldsymbol{\psi}(W_{s};\Xi)^{\top}]\mathbf{J}_{0}^{-1}$
	is the approximate covariance matrix, and $\mathbf{J}_{0}:=\partial E[\boldsymbol{\psi}(W_{s};\overline{\Xi})]/\partial\overline{\Xi}^{\top}|_{\overline{\Xi}=\Xi}$ is the Jacobian matrix. We now derive the approximate variances of
	$\hat{\phi}_{t}^{+}$ and $\hat{\phi}_{t}^{-}$ under the assumptions
	that $(t,\rho_{Y},\rho_{D},C_{\alpha},C_{Y},C_{\alpha})$ are all
	fixed. Since $\phi_{t}^{+}$ and $\phi_{t}^{-}$ are linear functions
	of $(\lambda^{+},\gamma^{+},\gamma^{-})$ and $(\lambda^{-},\gamma^{+},\gamma^{-})$, the influence functions (IFs) of $\phi_{t}^{+}$ and $\phi_{t}^{-}$
	are given by: 
	\begin{eqnarray*}
		\psi_{\phi_{t}^{+}} & = & \psi_{\lambda^{+}}-\psi_{\gamma^{-}}t\mathbf{1}\{t\geq0\}-\psi_{\gamma^{+}}t\mathbf{1}\{t<0\}\\
		& = & \psi_{\lambda s}+\zeta_{Y,\alpha}\psi_{S_{Y}}-(\psi_{\gamma_{s}}-\zeta_{D,\alpha}\psi_{S_{D}})t\mathbf{1}\{t\geq0\}-(\psi_{\gamma_{s}}+\zeta_{D,\alpha}\psi_{S_{D}})t\mathbf{1}\{t<0\}\\
		& = & \psi_{\lambda s}-\psi_{\gamma_{s}}t+(\zeta_{Y,\alpha}\psi_{S_{Y}}+\zeta_{D,\alpha}\psi_{S_{D}}|t|)\\
		& = & \mathbf{C}_{t}^{+\top}\boldsymbol{\psi}(W_{s};\Xi),
	\end{eqnarray*}
	and 
	\begin{eqnarray*}
		\psi_{\phi_{t}^{-}} & = & \psi_{\lambda^{-}}-\psi_{\gamma^{+}}t\mathbf{1}\{t\geq0\}-\psi_{\gamma^{-}}t\mathbf{1}\{t<0\}\\
		& = & \psi_{\lambda s}-\zeta_{Y,\alpha}\psi_{S_{Y}}-(\psi_{\gamma_{s}}+\zeta_{D,\alpha}\psi_{S_{D}})t\mathbf{1}\{t\geq0\}-(\psi_{\gamma_{s}}-\zeta_{D,\alpha}\psi_{S_{D}})t\mathbf{1}\{t<0\}\\
		& = & \psi_{\lambda s}-\psi_{\gamma_{s}}t-(\zeta_{Y,\alpha}\psi_{S_{Y}}+\zeta_{D,\alpha}\psi_{S_{D}}|t|)\\
		& = & \mathbf{C}_{t}^{-\top}\boldsymbol{\psi}(W_{s};\Xi),
	\end{eqnarray*}
	where 
	\begin{eqnarray*}
		\mathbf{C}_{t}^{+} & = & \left[1,-t,\left(\frac{\zeta_{Y,\alpha}\sigma_{Y_{s}}}{2v_{s}}+\frac{\zeta_{D,\alpha}\sigma_{D_{s}}|t|}{2v_{s}}\right),\frac{\zeta_{Y,\alpha}v_{s}}{2\sigma_{Y_{s}}},\frac{\zeta_{D,\alpha}v_{s}|t|}{2\sigma_{D_{s}}}\right]^{\top},\\
		\mathbf{C}_{t}^{-} & = & \left[1,-t,-\left(\frac{\zeta_{Y,\alpha}\sigma_{Y_{s}}}{2v_{s}}+\frac{\zeta_{D,\alpha}\sigma_{D_{s}}|t|}{2v_{s}}\right),-\frac{\zeta_{Y,\alpha}v_{s}}{2\sigma_{Y_{s}}},-\frac{\zeta_{D,\alpha}v_{s}|t|}{2\sigma_{D_{s}}}\right]^{\top}.
	\end{eqnarray*}
	
	It also can be shown that $\phi_{t}^{+}=\mathbf{C}_{t}^{+}\Xi$ and $\phi_{t}^{-}=\mathbf{C}_{t}^{-}\Xi$.
	With these results, the approximate variances of $\sqrt{n}(\hat{\phi}_{t}^{+}-\phi_{t}^{+})$
	and $\sqrt{n}(\hat{\phi}_{t}^{-}-\phi_{t}^{-})$ are $\mathbf{C}_{t}^{+\top}\boldsymbol{\Omega}\mathbf{C}_{t}^{+}$
	and $\mathbf{C}_{t}^{-\top}\boldsymbol{\Omega}\mathbf{C}_{t}^{-}$.
	If Assumption 4.2 (for PLIVM) or Assumption 5.2 (for LATE) in \citet{CCDDHNR_2018}
	holds, then we have: 
	\[
	\sqrt{n}(\hat{\phi}_{t}^{+}-\phi_{t}^{+})\overset{d}{\rightarrow}N\left(0,\mathbf{C}_{t}^{+\top}\boldsymbol{\Omega}\mathbf{C}_{t}^{+}\right),\sqrt{n}(\hat{\phi}_{t}^{-}-\phi_{t}^{-})\overset{d}{\rightarrow}N\left(0,\mathbf{C}_{t}^{-\top}\boldsymbol{\Omega}\mathbf{C}_{t}^{-}\right).
	\]\end{proof}
\subsection{Proof of Theorem 3}
\begin{proof}
	Using the result in Theorem 2, we immediately have
	$\lim_{n\rightarrow\infty}P(\hat{\phi}_{t,1-\tau}^{+}\geq\phi_{t}^{+})\geq1-\tau$
	and $\lim_{n\rightarrow\infty}P(\hat{\phi}_{t,\tau}^{-}\leq\phi_{t}^{-})\geq1-\tau$. Since $\phi_{t}\in[\phi_{t}^{+},\phi_{t}^{-}]$, the following also hold:
	\begin{equation}
		\lim_{n\rightarrow\infty}P(\hat{\phi}_{t,1-\tau}^{+}\geq\phi_{t}^{+})\geq1-\tau,
		\lim_{n\rightarrow\infty}P(\hat{\phi}_{t,\tau}^{-}\leq \phi_{t})\geq1-\tau.\label{pf_T3}
	\end{equation}
	If $\theta=\theta_{0}$, $\phi_{\theta_{0}}=0$ and (\ref{pf_T3}) becomes 
	\[
	\lim_{n\rightarrow\infty}P(\hat{\phi}_{t,1-\tau}^{+}\geq 0)\geq 1-\tau, \lim_{n\rightarrow\infty}P(\hat{\phi}_{t,\tau}^{-}\leq 0)\geq1-\tau,
	\]which is equivalent to 
	\[
	P\left(\theta_{0}\in\left\{t\in\boldsymbol{\Theta}_0: \hat{\phi}_{t,1-\tau}^{+}\geq0 \right\} \right) \geq 1-\tau,	P\left(\theta_{0}\in\left\{t\in\boldsymbol{\Theta}_0:\hat{\phi}_{t,\tau}^{-}\leq0\right\} \right)\geq 1-\tau.
	\]
\end{proof}
\subsection{Proof of Theorem 4}
\begin{proof}
	For $[\lambda^{+},\lambda^{-}]$, since $\lambda^{-}\leq\lambda^{+}$
	\begin{align*}
		P([\lambda^{+},\lambda^{-}]\in\text{CI}_{1-2\tau}^{[\lambda^{-},\lambda^{+}]})= & P(\{ \lambda^{-}\geq\hat{\lambda}_{\tau}^{-}\} \cap\{ \lambda^{+}\leq\hat{\lambda}_{1-\tau}^{+}\})\\
		= & 1-P(\{ \lambda^{-}<\hat{\lambda}_{\tau}^{-}\} \cup\{ \lambda^{+}>\hat{\lambda}_{1-\tau}^{+}\})\\
		\geq & 1-P(\lambda^{-}<\hat{\lambda}_{\tau}^{-})-P(\lambda^{+}>\hat{\lambda}_{1-\tau}^{+})\\
		= & 1-(1-P(\lambda^{-}\geq\hat{\lambda}_{\tau}^{-}))-(1-P(\lambda^{+}\leq\hat{\lambda}_{1-\tau}^{+}))\\
		\rightarrow & 1-2\tau
	\end{align*}as $n\rightarrow\infty$ by using the one-sided covering properties
	in \citet{CCNSS_2024}. The same argument can be applied to
	prove the case of $[\gamma^{-},\gamma^{+}]$. For $[\phi_{t}^{-},\phi_{t}^{+}]$,
	invoking Theorem 3 and using similar argument above yields:
	\begin{align*}
		P\left(\left[\hat{\phi}_{t,\tau}^{-},\hat{\phi}_{t,1-\tau}^{+}\right]\in\text{CI}_{1-2\tau}^{[\phi_{t}^{-},\phi_{t}^{+}]}\right)\geq & 1-P\left(\phi_{t}^{-}<\hat{\phi}_{t,\tau}^{-}\right)-P\left(\phi_{t}^{+}>\hat{\phi}_{t,1-\tau}^{+}\right)\\
		= & 1-\left(1-P\left( \phi_{t}^{-}\geq\hat{\phi}_{t,\tau}^{-} \right)\right)-\left(1-P\left( \phi_{t}^{-}\leq\hat{\phi}_{t,\tau}^{-}\right)\right)\\
		\rightarrow & 1-2\tau
	\end{align*}
	as $n\rightarrow\infty$. Furthermore, since $\lambda\in[\lambda^{+},\lambda^{-}]$,
	\[
	P(\lambda^{-}\geq\hat{\lambda}_{\tau}^{-})\leq P(\lambda\geq\hat{\lambda}_{\tau}^{-}), P(\lambda^{+}\leq\hat{\lambda}_{1-\tau}^{+})\leq P(\lambda\leq\hat{\lambda}_{1-\tau}^{+}).
	\]We can conclude that
	$P([\lambda^{-},\lambda^{+}]\in\text{CI}_{1-2\tau}^{[\lambda^{-},\lambda^{+}]})\leq P(\lambda\in\text{CI}_{1-2\tau}^{[\lambda^{-},\lambda^{+}]})$
	(see also Lemma 1 of \citet{IM_2004}). A similar result holds for
	$\gamma$ and $\phi_{t}$. 
\end{proof}

\section{Derivations for $E[g_{Ys}(\alpha-\alpha_s)]=0$ and $E[g_{Ds}(\alpha-\alpha_s)]=0$}
Key conditions to arrive (\ref{OVB_lambda}) and (\ref{OVB_gamma}) are $E[g_{Ys}(\alpha-\alpha_{s})]=0$
and $E[g_{Ds}(\alpha-\alpha_{s})]=0$. It is easy to see that if $E[\alpha(W)|W_{s}]=\alpha_{s}(W_s)$, the two key conditions also hold. $E[\alpha(W)|W_{s}]=\alpha_{s}$ holds
for LATE and LATT. For LATE, 
\[
\alpha(W)=\frac{Z}{\pi(X,A)}-\frac{1-Z}{1-\pi(X,A)},\alpha_{s}(W_{s})=\frac{Z}{\pi(X)}-\frac{1-Z}{1-\pi(X)}.
\]
The first term of $E[\alpha(W)|W_s]$ is
\begin{align}
	E\left[\frac{Z}{\pi(X,A)}|W_{s}\right] & =E\left[\frac{ZP(X,A)}{P(Z=1,X,A)}|Z,X\right].\label{LATE_alpha}
\end{align}
When $Z=1$, 
\begin{align*}
	E\left[\frac{ZP(X,A)}{P(Z=1,X,A)}|Z=1,X\right] & =\int\frac{P(X,a)}{P(Z=1,X,a)}f_{A}(a|Z=1,X)da=\frac{1}{\pi(X)},
\end{align*}which equals to the first term of $\alpha_s(W_s)$ when $Z=1$. When $Z=0$, the conditional expectation (\ref{LATE_alpha}) and $Z/\pi(X)$ are both zero. Therefore we can conclude that 
\[
E\left[\frac{Z}{\pi(X,A)}|W_{s}\right]=\frac{Z}{\pi(X)}.
\]
Using similar arguments, for the second term of $E[\alpha(W)|W_s]$, we also have
\[
E\left[\frac{1-Z}{1-\pi(X,A)}|W_{s}\right]=\frac{1-Z}{1-\pi(X)}.
\]
Therefore we conclude $E[\alpha(W)|W_{s}]=\alpha_{s}(W_{s})$ holds
for LATE. 

For LATT, 
\[
\alpha(W)=Z-\frac{\pi(X,A)}{1-\pi(X,A)}(1-Z),\alpha_{s}(W_{s})=Z-\frac{\pi(X)}{1-\pi(X)}(1-Z).
\]
Then
\begin{align*}
	E[\alpha(W)|W_{s}] & =Z-E\left[\frac{\pi(X,A)}{1-\pi(X,A)}(1-Z)|Z,X\right].
\end{align*}
If $Z=1$, $E[\alpha(W)|Z=1,X]=1=\alpha_{s}(W_{s})$. If $Z=0$, 
\begin{align*}
	E[\alpha(W)|W_{s}] & =-E\left[\frac{\pi(X,A)}{1-\pi(X,A)}|Z=0,X\right]\\
	& =-\int\frac{P(Z=1,X,a)}{P(Z=0,X,a)}f_{A|Z,X}(a|Z=0,X)da\\
	& =-\int\frac{P(Z=1,X,a)}{P(Z=0,X)}da\\
	& =-\frac{\pi(X)}{1-\pi(X)}=\alpha_{s}(W_{s}).
\end{align*}
Therefore we conclude that $E[\alpha(W)|W_{s}]=\alpha_{s}(W_s)$ holds
for LATT.

However, $E[\alpha|W_{s}]=\alpha_{s}$ does not hold for PLIVM. To
show that $E[g_{Ys}(\alpha-\alpha_{s})]=0$ and $E[g_{Ds}(\alpha-\alpha_{s})]=0$
still hold for PLIVM, we directly calculate these expectations.
Following \citet{CCNSS_2024}, we define the short version of $g_Y(W)$ as
\[
g_{Ys}(W_s)  :=  \lambda_{s}Z+k_{s}(X),
\]
where $k_{s}(X)=\theta_{s}h_{s}(X)+f_{s}(X)$. Note that
\[
\alpha-\alpha_{s} =\frac{\sigma_{Z_{s}}^{2}(Z-E[Z|X,A])-\sigma_{Z}^{2}(Z-E[Z|X])}{\sigma_{Z}^{2}\sigma_{Z_{s}}^{2}},
\]where 
$\sigma_{Z}^{2}=E[(Z-E[Z|X,A])^{2}]$ and $\sigma_{Z_{s}}^{2}=E[(Z-E[Z|X])^{2}]$.
Next,
\begin{align*}
	E[Z(\alpha-\alpha_{s})]= & \frac{\sigma_{Z_{s}}^{2}E[Z(Z-E[Z|X,A])]-\sigma_{Z}^{2}E[Z(Z-E[Z|X])]}{\sigma_{Z}^{2}\sigma_{Z_{s}}^{2}}\\
	= & \frac{\sigma_{Z_{s}}^{2}E[Z^{2}-(E[Z|X,A])^{2}]-\sigma_{Z}^{2}E[Z^{2}-(E[Z|X])^{2}]}{\sigma_{Z}^{2}\sigma_{Z_{s}}^{2}}\\
	= & \frac{\sigma_{Z_{s}}^{2}\sigma_{Z}^{2}-\sigma_{Z}^{2}\sigma_{Z_{s}}^{2}}{\sigma_{Z}^{2}\sigma_{Z_{s}}^{2}} = 0.
\end{align*}Also,
\begin{align*}
	E[k_{s}(X)(\alpha-\alpha_{s})]= & \frac{\sigma_{Z_{s}}^{2}E[k_{s}(X)(Z-E[Z|X,A])]-\sigma_{Z}^{2}E[k_{s}(X)(Z-E[Z|X])]}{\sigma_{Z}^{2}\sigma_{Z_{s}}^{2}}\\
	= & \frac{\sigma_{Z_{s}}^{2}E[E[k_{s}(X)(Z-E[Z|X,A])|X]]}{\sigma_{Z}^{2}\sigma_{Z_{s}}^{2}}-\\
	& \frac{\sigma_{Z}^{2}E[E[k_{s}(X)(Z-E[Z|X])|X]]}{\sigma_{Z}^{2}\sigma_{Z_{s}}^{2}}\\
	= & \frac{\sigma_{Z_{s}}^{2}E[k_{s}(X)E[(Z-E[Z|X,A])|X]]}{\sigma_{Z}^{2}\sigma_{Z_{s}}^{2}}-\\
	& \frac{\sigma_{Z}^{2}E[k_{s}(X)E[(Z-E[Z|X])|X]]}{\sigma_{Z}^{2}\sigma_{Z_{s}}^{2}}\\
	= & 0.
\end{align*}
Combining the above results, we conclude that $E[g_{Ys}(\alpha-\alpha_{s})]=0$
for PLIVM. Similar arguments can be applied for proving $E[g_{Ds}(\alpha-\alpha_{s})]=0$ with $g_{Ds}(W_s):=\gamma_s Z +h_s(X)$. Finally, in this case, $E[(\alpha-\alpha_s)^2] = E[\alpha^2] - E[\alpha_{s}^2]$ also holds since
\begin{align*}
	E[\alpha_{s}(\alpha-\alpha_{s})]	&=E\left[\frac{Z-E[Z|X]}{\sigma_{Z_{s}}^{2}}\times\frac{\sigma_{Z_{s}}^{2}(Z-E[Z|X,A])-\sigma_{Z}^{2}(Z-E[Z|X])}{\sigma_{Z}^{2}\sigma_{Z_{s}}^{2}}\right]\\
	&=\frac{\sigma_{Z_{s}}^{2}E[(Z-E[Z|X])(Z-E[Z|X,A])]-\sigma_{Z}^{2}\sigma_{Z_{s}}^{2}}{\sigma_{Z}^{2}\sigma_{Z_{s}}^{4}}\\
	&=\frac{\sigma_{Z_{s}}^{2}E[Z-(E[Z|X,A])^{2}]-\sigma_{Z}^{2}\sigma_{Z_{s}}^{2}}{\sigma_{Z}^{2}\sigma_{Z_{s}}^{4}}\\
	&=\frac{\sigma_{Z_{s}}^{2}\sigma_{Z}^{2}-\sigma_{Z}^{2}\sigma_{Z_{s}}^{2}}{\sigma_{Z}^{2}\sigma_{Z_{s}}^{4}}=0
\end{align*}
\clearpage
\begin{table}
	\caption{Maximum estimates of $C_{\alpha}$, $C_{Y}$ and $C_{D}$ using the benchmark analysis with $k_{\alpha}= k_Y = k_D = 1$}
	\centering %
	\begin{tabular}{rcccccc}
		\hline 
		&  & \multicolumn{2}{c}{LATE} &  & \multicolumn{2}{c}{LATT}\tabularnewline
		\cline{3-4} \cline{4-4} \cline{6-7} \cline{7-7} 
		&  & $X_{j}^{*}$  & Est. &  & $X_{j}^{*}$  & Est.\tabularnewline
		\hline 
		& $C_{\alpha}$  & \texttt{age}  & 0.138 &  & \texttt{age}  & 0.195\tabularnewline
		Male  & $C_{Y}$  & \texttt{wkless13}  & 0.147 &  & \texttt{wkless13}  & 0.147\tabularnewline
		& $C_{D}$  & \texttt{hsorged}  & 0.043 &  & \texttt{hsorged}  & 0.043\tabularnewline
		&  &  &  &  &  & \tabularnewline
		& $C_{\alpha}$  & \texttt{wkless13}  & 0.079 &  & \texttt{wkless13}  & 0.106\tabularnewline
		Female  & $C_{Y}$  & \texttt{wkless13}  & 0.181 &  & \texttt{wkless13}  & 0.181\tabularnewline
		& $C_{D}$  & \texttt{hsorged}  & 0.059 &  & \texttt{hsorged}  & 0.059\tabularnewline
		\hline 
	\end{tabular}\label{table1} 
\end{table}

\begin{table}
	\centering 
	\caption{OVB analysis results of the LATE for male and female workers. For both groups, we set $|\rho_{Y}|=|\rho_{D}|=1$. The maximum estimates of the sensitivity parameters shown in Table \ref{table1} are calibrated to generate the result. 
	}
	\begin{tabular}{lccccccc}
		\hline 
		& \multicolumn{7}{c}{Male}\tabularnewline
		\cline{2-8} \cline{3-8} \cline{4-8} \cline{5-8} \cline{6-8} \cline{7-8} \cline{8-8} 
		& Est.  &  & C.I. (95\%)  &  & OVB Bound Est.  &  & $[\text{Low}_{0.025},\text{Up}_{0.975}]$\tabularnewline
		\hline 
		$\theta_0$ (LATE)& 1,664.55  &  & [-186.90, 3,516.00]  &  & [317.62, 3,043.88]  &  & [-1,551.26, 4,916.20]\tabularnewline
		$\lambda$  & 1,023.02  &  & [-114.99, 2,161.03]  &  & [196.40, 1,849.64]  &  & [-940.76, 2989.13]\tabularnewline
		$\gamma$  & 0.61  &  & [0.60, 0.63]  &  & [0.61, 0.62]  &  & [0.59, 0.64]\tabularnewline
		\hline 
		& \multicolumn{7}{c}{Female}\tabularnewline
		\cline{2-8} \cline{3-8} \cline{4-8} \cline{5-8} \cline{6-8} \cline{7-8} \cline{8-8} 
		& Est.  &  & C.I. (95\%)  &  & OVB Bound Est. &  & $[\text{Low}_{0.025},\text{Up}_{0.975}]$\tabularnewline
		\hline 
		$\theta_0$ (LATE) & 1,900.10  &  & [816.14, 2,984.05]  &  & [1,279.58, 2,530.10]  &  & [198.50, 3,625.84]\tabularnewline
		$\lambda$  & 1,231.73  &  & [525.99, 1,937.47]  &  & [834.50, 1,628.96]  &  & [129.28, 2,335.42]\tabularnewline
		$\gamma$  & 0.65  &  & [0.63, 0.66]  &  & [0.64, 0.65]  &  & [0.63, 0.67]\tabularnewline
		\hline 
	\end{tabular}\label{table2} 
\end{table}

\begin{table}
	\caption{OVB analysis results of the LATT for male and female workers. For both groups, we set $|\rho_{Y}|=|\rho_{D}|=1$. The maximum estimates of the sensitivity parameters shown in Table \ref{table1} are calibrated to generate the result. 
	}
	\centering %
	\begin{tabular}{lccccccc}
		\hline 
		& \multicolumn{7}{c}{Male }\tabularnewline
		\cline{2-8} \cline{3-8} \cline{4-8} \cline{5-8} \cline{6-8} \cline{7-8} \cline{8-8} 
		& Est.  &  & C.I. (95\%)  &  & OVB Bound Est. &  & $[\text{Low}_{0.025},\text{Up}_{0.975}]$\tabularnewline
		\hline 
		$\theta_0$ (LATT) & 1,634.10  &  & [-270.03, 3,538.24]  &  & [-300.71, 3,606.59]  &  & [-2,244.03, 5,546.17]\tabularnewline
		$\lambda$  & 1,002.25  &  & [-173.57, 2,178.08]  &  & [-182.34, 2,186.85]  &  & [-1,358.01, 3,364.23]\tabularnewline
		$\gamma$  & 0.61  &  & [0.60, 0.63]  &  & [0.61, 0.62]  &  & [0.59, 0.64]\tabularnewline
		\hline 
		& \multicolumn{7}{c}{Female}\tabularnewline
		\cline{2-8} \cline{3-8} \cline{4-8} \cline{5-8} \cline{6-8} \cline{7-8} \cline{8-8} 
		& Est.  &  & C.I. (95\%)  &  & OVB Bound Est. &  & $[\text{Low}_{0.025},\text{Up}_{0.975}]$\tabularnewline
		\hline 
		$\theta_0$ (LATT) & 1,993.21  &  & [879.80, 3,106.62]  &  & [1,150.82, 2,851.74]  &  & [46.66, 3,977.10]\tabularnewline
		$\lambda$  & 1,292.02  &  & [569.19, 2,014.84]  &  & [752.25, 1,831.78]  &  & [30.45, 2,556.01]\tabularnewline
		$\gamma$  & 0.65  &  & [0.63, 0.66]  &  & [0.64, 0.65]  &  & [0.63, 0.67]\tabularnewline
		\hline 
	\end{tabular}\label{table3} 
\end{table}
\clearpage
\begin{table}\caption{Results of the OVB-adjusted confidence intervals of the LATE for male and female workers. The maximum estimates of the sensitivity parameters shown in Table \ref{table1} are calibrated to generate the result. For $\theta_0$ (LATE), $z_{l}^{*}$, $z_{u}^{*}$, $\hat{\Delta}^{*}$ and Min. Obj. are shown in averages obtained from solving the constrained minimization problem of Stoye's shrinkage method for $\phi_t$.}
\centering	
\begin{tabular}{lccccccccc}
	\hline 
	&  & \multicolumn{8}{c}{Male}\tabularnewline
	\cline{3-10} \cline{4-10} \cline{5-10} \cline{6-10} \cline{7-10} \cline{8-10} \cline{9-10} \cline{10-10} 
	&  & OVB-adj. C.I. (95\%) &  & $z_{l}^{*}$ & $z_{u}^{*}$ &  & $\hat{\Delta}^{*}$ &  & Min. Obj.\tabularnewline
	\hline 
	$\theta_0$ (LATE) &  & [-1,249.34, 4,615.08] &  & 1.64 & 1.64 &  & 1,679.98 &  & 136,610.49\tabularnewline
	$\lambda$ &  & [-757.95, 2805.95] &  & 1.64 & 1.64 &  & 1,653.25 &  & 136,474.56\tabularnewline
	$\gamma$ &  & [0.59, 0.64] &  & 1.96 & 1.96 &  & 0.00 &  & 2.52\tabularnewline
	\hline 
	&  & \multicolumn{8}{c}{Female}\tabularnewline
	\cline{3-10} \cline{4-10} \cline{5-10} \cline{6-10} \cline{7-10} \cline{8-10} \cline{9-10} \cline{10-10} 
	&  & OVB-adj. C.I. (95\%) &  & $z_{l}^{*}$ & $z_{u}^{*}$ &  & $\hat{\Delta}^{*}$ &  & Min. Obj.\tabularnewline
	\hline 
	$\theta_0$ (LATE) &  & [372.18, 3,449.81] &  & 1.65 & 1.65 &  & 811.15 &  & 92,697.95\tabularnewline
	$\lambda$ &  & [242.46, 2,222.04] &  & 1.65 & 1.65 &  & 794.46 &  & 92,576.32\tabularnewline
	$\gamma$ &  & [0.63, 0.67] &  & 1.96 & 1.96 &  & 0.00 &  & 2.50\tabularnewline
	\hline 
\end{tabular}\label{table4}
\end{table}

\begin{table}\caption{Results of the OVB-adjusted confidence intervals of the LATT for male and female workers. The maximum estimates of the sensitivity parameters shown in Table \ref{table1} are calibrated to generate the result. For $\theta_0$ (LATT), $z_{l}^{*}$, $z_{u}^{*}$, $\hat{\Delta}^{*}$ and Min. Obj. are shown in averages obtained from solving the constrained minimization problem of Stoye's shrinkage method for $\phi_t$.}
\centering
\begin{tabular}{lccccccccc}
	\hline 
	&  & \multicolumn{8}{c}{Male}\tabularnewline
	\cline{3-10} \cline{4-10} \cline{5-10} \cline{6-10} \cline{7-10} \cline{8-10} \cline{9-10} \cline{10-10} 
	&  & OVB-adj. C.I. (95\%) &  & $z_{l}^{*}$ & $z_{u}^{*}$ &  & $\hat{\Delta}^{*}$ &  & Min. Obj.\tabularnewline
	\hline 
	$\theta_0$ (LATT) &  & [-1,930.97, 5,234.14] &  & 1.64 & 1.64 &  & 2,415.14 &  & 141,246.66\tabularnewline
	$\lambda$ &  & [-1,168.99, 3,174.93] &  & 1.64 & 1.64 &  & 2,369.18 &  & 141,052.63\tabularnewline
	$\gamma$ &  & {[}0.59, 0.64{]} &  & 1.65 & 1.65 &  & 0.02 &  & 2.14\tabularnewline
	\hline 
	&  & \multicolumn{8}{c}{Female}\tabularnewline
	\cline{3-10} \cline{4-10} \cline{5-10} \cline{6-10} \cline{7-10} \cline{8-10} \cline{9-10} \cline{10-10} 
	&  & OVB-adj. C.I. (95\%) &  & $z_{l}^{*}$ & $z_{u}^{*}$ &  & $\hat{\Delta}^{*}$ &  & Min. Obj.\tabularnewline
	\hline 
	$\theta_0$ (LATT) &  & [137.51, 3884.57] &  & 1.64 & 1.64 &  & 1,103.24 &  & 94,963.22\tabularnewline
	$\lambda$ &  & [89.76, 2496.51] &  & 1.64 & 1.64 &  & 1,079.52 &  & 94,801.09\tabularnewline
	$\gamma$ &  & [0.62, 0.67] &  & 1.96 & 1.96 &  & 0.00 &  & 2.54\tabularnewline
	\hline 
\end{tabular}\label{table5}
\end{table}

\clearpage
\begin{figure}
	\includegraphics[width = 8cm, height = 8cm]{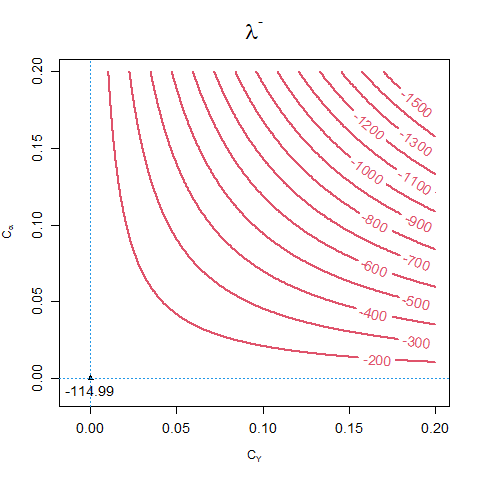}
	\includegraphics[width = 8cm, height = 8cm]{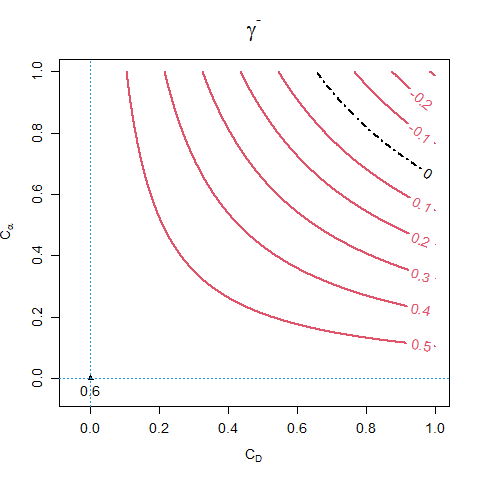}
	\caption{Sensitivity contour plots of $\lambda^{-}$ (left panel) and $\gamma^{-}$ (right panel) for the LATE of male workers. The figures show lower bounds of the $(1-\tau)$ confidence intervals for $\lambda^{-}$ and $\gamma^{-}$. We set $\tau = 0.025$ and $|\rho_Y|=|\rho_D|=1$. The cut-off point of $C_{\alpha}C_D$ is 0.66.}
	\label{figure1}
\end{figure}

\begin{figure}
	\centering
	\includegraphics[width = 8cm, height = 8cm]{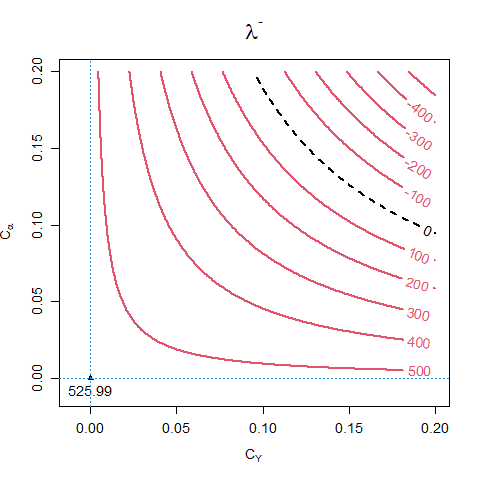}
	\includegraphics[width = 8cm, height = 8cm]{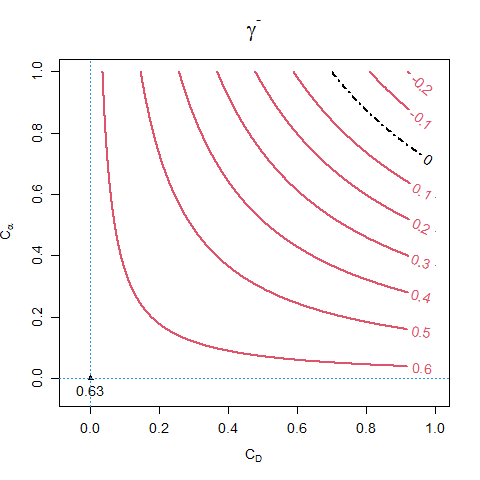}
	\caption{Sensitivity contour plots of $\lambda^{-}$ (left panel) and $\gamma^{-}$ (right panel) for the LATE of female workers. The figures show lower bounds of the $(1-\tau)$ confidence intervals for $\lambda^{-}$ and $\gamma^{-}$. We set $\tau = 0.025$ and $|\rho_Y|=|\rho_D|=1$. The cut-off points of $C_{\alpha}C_Y$ and $C_{\alpha}C_D$ are 0.019 and 0.701.}
	\label{figure2}
\end{figure}

\begin{figure}
	\includegraphics[width = 8cm, height = 8cm]{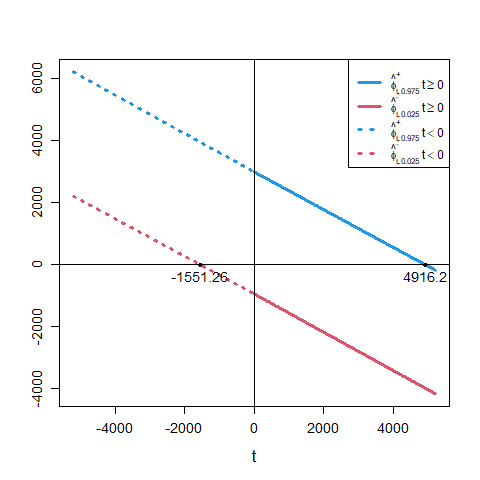}
	\includegraphics[width = 8cm, height = 8cm]{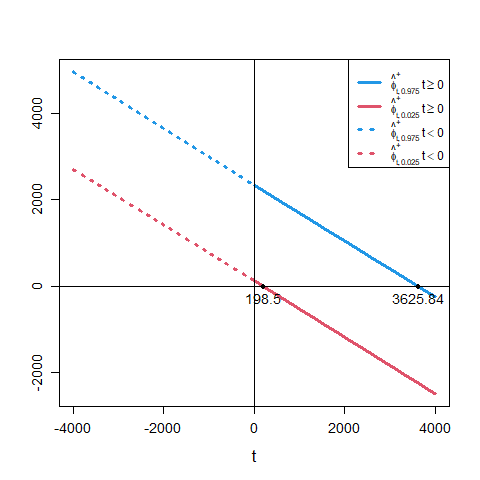}
	\caption{Plots of $\hat{\phi}_{t,1-\tau}^{+}$ and $\hat{\phi}_{t,\tau}^{-}$ for the LATE of male (left panel) and female (right panel) workers. We set $\tau = 0.025$ and $|\rho_{Y}|=|\rho_{D}|=1$. The maximum estimates of the sensitivity parameters shown in Table \ref{table1} are calibrated to generate the result. 
	}
	\label{figure3}
\end{figure}
\clearpage

\begin{figure}
	\includegraphics[width = 8cm, height = 8cm]{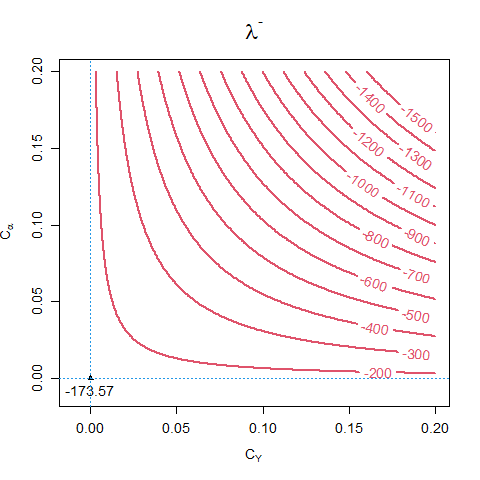}
	\includegraphics[width = 8cm, height = 8cm]{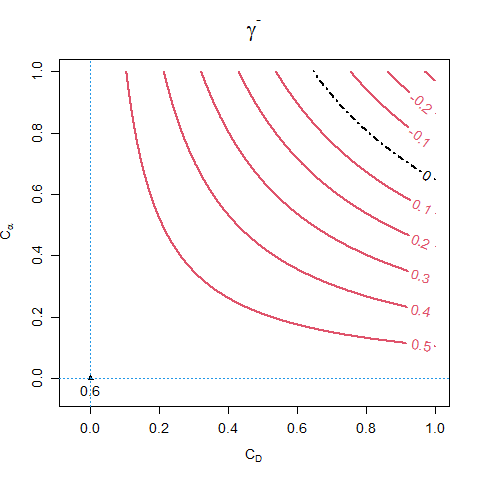}
	\caption{Sensitivity contour plots of $\lambda^{-}$ (left panel) and $\gamma^{-}$ (right panel) for the LATT of male workers. The figures show lower bounds of the $(1-\tau)$ confidence intervals for $\lambda^{-}$ and $\gamma^{-}$. We set $\tau = 0.025$ and $|\rho_Y|=|\rho_D|=1$. The cut-off point of 
		$C_{\alpha}C_D$ is 0.648}
\label{figure4}
\end{figure}

\begin{figure}
\includegraphics[width = 8cm, height = 8cm]{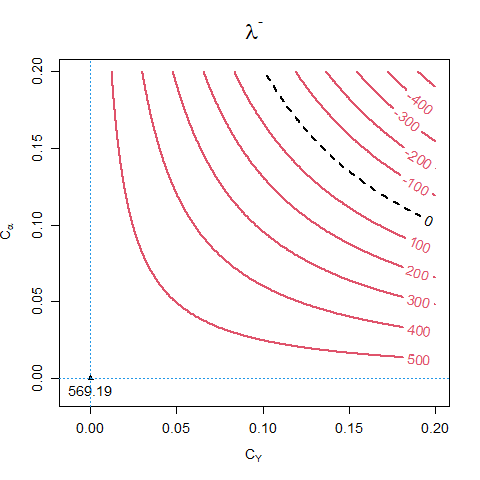}
\includegraphics[width = 8cm, height = 8cm]{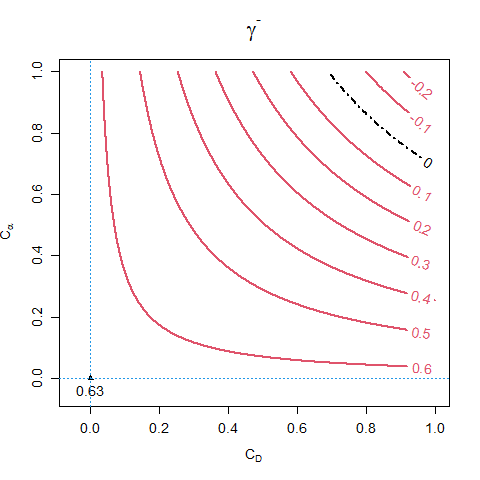}
\caption{Sensitivity contour plots of $\lambda^{-}$ (left panel) and $\gamma^{-}$ (right panel) for the LATT of female workers. The figures show lower bounds of the $(1-\tau)$ confidence intervals for $\lambda^{-}$ and $\gamma^{-}$. We set $\tau = 0.025$ and $|\rho_Y|=|\rho_D|=1$. The cut-off points of $C_{\alpha}C_Y$ and $C_{\alpha}C_D$ are 0.020 and 0.691.}
\label{figure5}
\end{figure}

\begin{figure}
\includegraphics[width = 8cm, height = 8cm]{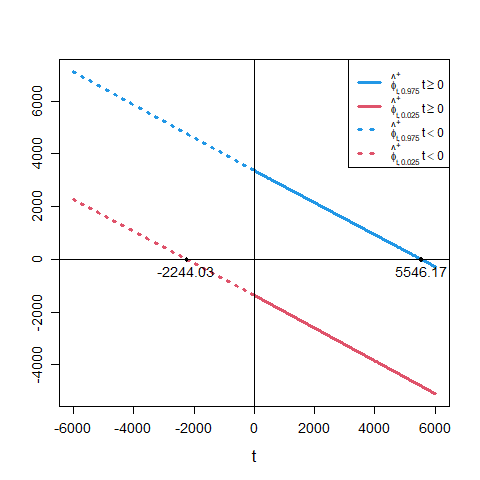}
\includegraphics[width = 8cm, height = 8cm]{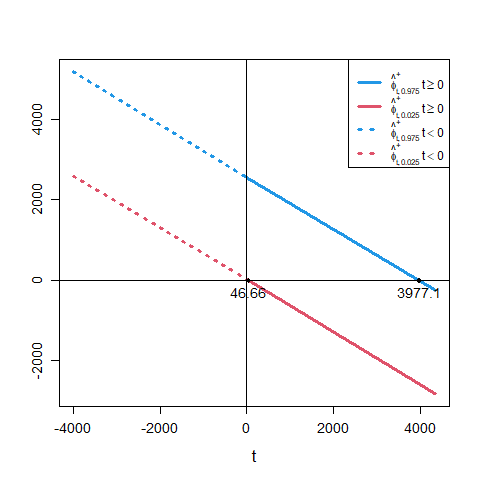}
\caption{Plots of $\hat{\phi}_{t,1-\tau}^{+}$ and $\hat{\phi}_{t,\tau}^{-}$ for the LATT of male (left) and female (right) workers. We set $\tau = 0.025$ and $|\rho_{Y}|=|\rho_{D}|=1$. The maximum estimates of the sensitivity parameters shown in Table \ref{table1} are calibrated to generate the result. 
}
\label{figure6}
\end{figure}
\clearpage

\begin{figure}
	\includegraphics[width = 8cm, height = 8cm]{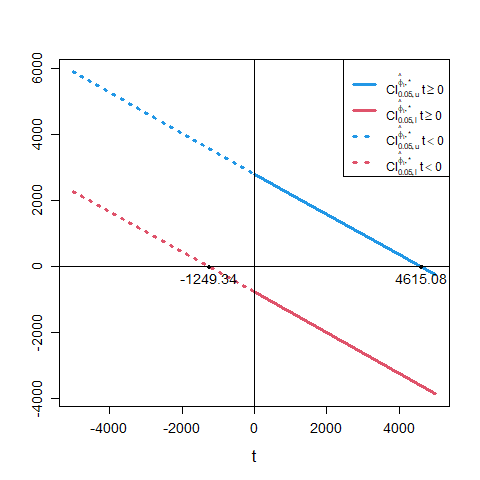}
	\includegraphics[width = 8cm, height = 8cm]{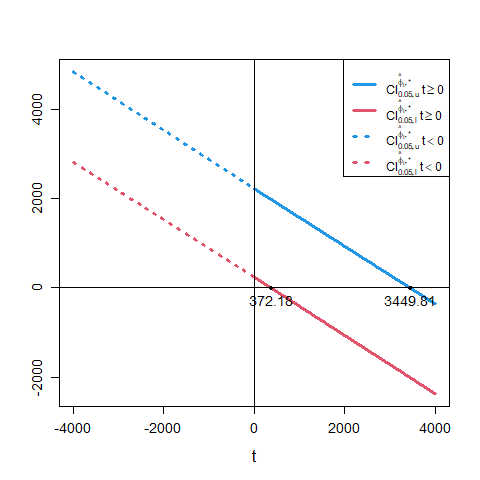}
	\caption{Plots of upper and lower bounds of $\text{CI}_{1-\tau}^{\phi_t,*}$ for the LATE of male (left) and female (right) workers. We set $\tau = 0.05$ and $|\rho_{Y}|=|\rho_{D}|=1$. The maximum estimates of the sensitivity parameters shown in Table \ref{table1} are calibrated to generate the result. 
	}
	\label{figure7}
\end{figure}
\begin{figure}
	\includegraphics[width = 8cm, height = 8cm]{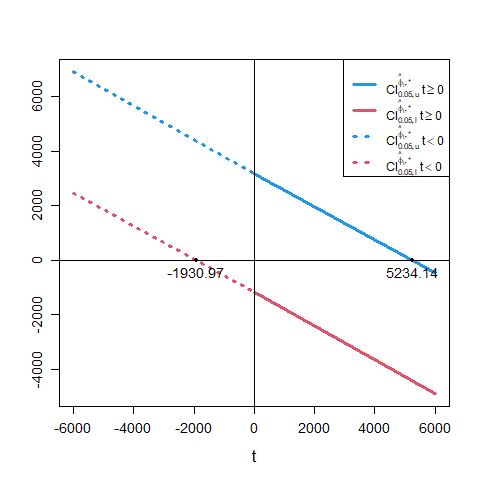}
	\includegraphics[width = 8cm, height = 8cm]{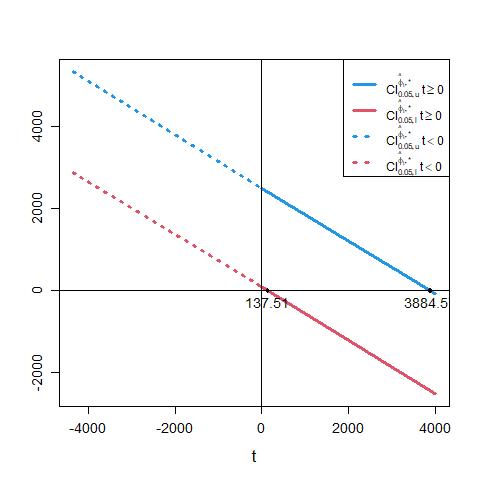}
	\caption{Plots of upper and lower bounds of $\text{CI}_{1-\tau}^{\phi_t,*}$ for the LATT of male (left) and female (right) workers. We set $\tau = 0.05$ and $|\rho_{Y}|=|\rho_{D}|=1$. The maximum estimates of the sensitivity parameters shown in Table \ref{table1} are calibrated to generate the result. 
	}
	\label{figure8}
\end{figure}
\clearpage
\bibliographystyle{ECTA}
\bibliography{ref_OVB}

\clearpage

\end{document}